\documentclass[11pt]{amsart}  
\usepackage[left=1.3in,right=1.3in,bottom=1.3in]{geometry}  
\usepackage[T1]{fontenc}

\usepackage{amssymb,amsmath,amscd}

\usepackage[parfill]{parskip}

\usepackage{appendix}
\usepackage{upgreek}
\usepackage{mathrsfs}
\usepackage{latexsym}

\usepackage{enumitem}
\setlist{itemsep=0.2em, parsep=0.2em} 
\setlist[itemize,1]{label=\ensuremath{\blacktriangleright}}
\setlist[itemize,2]{label=\ensuremath{\triangleright}}

\usepackage{graphicx} 
\usepackage{tikz-cd}
\usepackage{stmaryrd}

\theoremstyle{plain}
\newtheorem{theorem}{Theorem}[section]
\newtheorem{lemma}[theorem]{Lemma}
\newtheorem{proposition}[theorem]{Proposition}
\newtheorem{corollary}[theorem]{Corollary}

\theoremstyle{definition}
\newtheorem{definition}[theorem]{Definition}
\newtheorem{example}[theorem]{Example}
\newtheorem{remark}[theorem]{Remark}

\newcommand{\id}{\textsc{id}} 
\newcommand{\Fix}{\textsc{Fix}} 
\newcommand{\inagg}{\alpha} 
\newcommand{\distrib}{\delta} 
\newcommand{\outagg}{\omega} 
\newcommand{\liab}{\ell} 

\newcommand{\hyper}{X} 
\newcommand{\source}[1]{s(#1)} 
\newcommand{\target}[1]{t(#1)} 
\newcommand{\sheaf}[1]{\mathcal{#1}} 
\newcommand{\stalk}[2]{{#1}(#2)} 
\newcommand{\Inc}{\mathcal{I}}
\DeclareMathOperator{\Hom}{Hom}
\DeclareMathOperator{\Eq}{Eq}

\newcommand{\hypedg}{h} 
\newcommand{\dis}{\mathrm{dis}}   
\newcommand{\col}{\mathrm{col}}   
\newcommand{\datacat}{\mathbf{C}} 
\newcommand{\payob}{P} 
\newcommand{\Phidual}{\Phi^\sharp}  
\newcommand{\Phidualast}{\Phi^\sharp_\ast}  
\newcommand{\style}[1]{\emph{#1}}

\begin{document}

\title{Clearing in Liability Networks via Sheaves on Directed Hypergraphs}
\author{Robert Ghrist}
\address{Department of Mathematics and Electrical \& Systems Engineering \\
         University of Pennsylvania \\
         Philadelphia, PA}
\begin{abstract}
We associate to a decorated liability network a \emph{liability sheaf} on a directed hypergraph whose hyperedges separate the distribution of payments from the collection of receipts. Clearing configurations are precisely the global sections of this sheaf, and the global-section object is canonically the equalizer of the identity and a clearing operator $\Phi=A\circ D$ factored into collective distribution $D$ and aggregation $A$; an \emph{institution-edge duality} identifies it equivalently with the equalizer of the dual operator $D\circ A$ on the edge side. This identifies liability clearing as a finite-limit construction in the ambient data category. The construction is functorial under change of coefficient category: a \emph{Clearing Invariance Theorem} shows that a finite-limit-preserving functor compatible with constraint subobjects induces a canonical isomorphism on global-section objects, enabling uniform comparison of clearing problems across categories of payment data. Existence, uniqueness, and iterative computation of clearing sections are organized by the structure carried on payment objects: Tarski's theorem yields existence and a complete-lattice structure under complete-lattice global elements; Scott continuity refines this to convergent Kleene iteration; an acyclic underlying graph admits a unique clearing section in finitely many steps with no order or metric hypothesis; and Banach's theorem on global elements yields uniqueness under metric contraction. The Eisenberg--Noe model and lattice liability networks arise as special cases.
\end{abstract}


\maketitle

\section{Introduction}
\label{sec:intro}

The Eisenberg-Noe model \cite{EisenbergNoe2001} established the mathematical foundations for analyzing financial clearing through fixed points of payment flows on directed networks. There followed numerous extensions and reformulations: see Section \ref{subsec:related}. Recent work \cite{GhristGouldLopezRiess2025} generalized this framework to lattice liability networks (LLNs), where payments take values in complete lattices rather than real numbers, and clearing configurations satisfy local balance conditions reminiscent of global sections of a sheaf, though the relationship between local constraints and global clearing remained implicit.

This paper recasts liability clearing as a problem in sheaf theory. We construct liability sheaves on directed hypergraphs where clearing configurations correspond precisely to global sections. The hypergraph structure separates the concerns of resource distribution and collection: distribution hyperedges model how institutions allocate payments among creditors, while collection hyperedges aggregate receipts from debtors. Individual liabilities remain edge-indexed as in classical models, but the hypergraph topology clarifies how these obligations compose into network-wide payment flows.

The sheaf framework requires working in a liability category $\datacat$ satisfying specific conditions: finite limits, ordered global elements, distinguished constraint subobjects, and a bound selector relating liabilities to constraints. When $\datacat$ is the category of partially ordered sets, we recover both Eisenberg-Noe and LLN models as special cases. The choice of coefficient category determines the structural form of clearing; additional order or metric structure on the set of global elements determines the applicable fixed-point principle (Tarski's theorem for complete-lattice global elements; Banach's theorem when the global elements of the total payment object carry a complete metric and the clearing operator contracts).

The remainder of the introduction surveys the related literature (\S\ref{subsec:related}), records the categorical and order-theoretic conventions in force throughout (\S\ref{subsec:math}), and summarizes the contributions (\S\ref{subsec:summary}).

\subsection{Related Work}
\label{subsec:related}

The modern theory of network clearing originates with Eisenberg and Noe, who cast clearing as a monotone self-map on a product lattice of payment intervals and invoke Tarski's theorem for existence of a greatest clearing vector \cite{EisenbergNoe2001}. The essential ingredient is monotonicity, not linearity or convexity. The lattice structure and uniqueness of clearing payment matrices under general bankruptcy rules are studied in \cite{CsokaHerings2024}. Order-theoretic extensions retain this fixed-point core while adding financial realism: Rogers--Veraart incorporate default costs preserving order-theoretic existence \cite{RogersVeraart2013}; Elsinger treats cross-holdings coupling equity and debt channels \cite{Elsinger2009}; Elliott, Golub, and Jackson analyze contagion under feedback through liabilities and cross-exposures \cite{ElliottGolubJackson2014}. The \style{lattice liability network} framework of \cite{GhristGouldLopezRiess2025} generalizes the product-lattice setting to arbitrary complete-lattice vertex state spaces and separates pay-in aggregation from pay-out distribution via monotone operators; this is the immediate antecedent of the present work, and we recover it as a special case.

Sheaf-theoretic perspectives on networks provide a local-to-global language for consistency. Cellular and network sheaves with their Laplacians supply tools for acyclic flows \cite{HansenGhrist2019,HansenGhrist2021}; quiver-based operators adapt these to directed data \cite{SumrayHarringtonNanda2024}. Duta, Cassar\`a, Silvestri, and Li\`o introduce sheaf hypergraph Laplacians \cite{DutaCassaraSilvestriLio2023}, and a directional variant with corrected spectral properties appears in \cite{DirectionalSHN2025}. These constructions target diffusion and learning rather than clearing, and cyclic liabilities undermine Laplacian descent in either setting; the present paper takes an order-theoretic (rather than variational) route.

Further neighboring lines -- dynamic, multi-period, and optimized clearing \cite{CapponiChen2015,BanerjeeBernsteinFeinstein2018,BanerjeeFeinstein2019,BarrattBoyd2020,CalafioreEtAl2023,CalafioreFracastoroProskurnikov2024,Feinstein2019,KusnetsovVeraart2019}, information incompleteness and network reconstruction \cite{AnandCraigVonPeter2015,GandyVeraart2017}, and algorithmic complexity \cite{BestingHoeferHuth2026,SchuldenzuckerSeukenBattiston2017} -- remain complementary rather than subsumed. Broader systemic-risk and amplification questions appear in \cite{AcemogluOzdaglarTahbazSalehi2015,GlassermanYoung2015,GaiHaldaneKapadia2011}. This paper sharpens the sheaf-theoretic direction: clearing becomes a finite-limit construction, functorial in the coefficient category.

\subsection{Mathematical tools}
\label{subsec:math}

Throughout, $\datacat$ denotes a \style{liability category} (Definition~\ref{def:liability-category}): a category with finite limits, ordered global elements, a pullback-stable class of constraint subobjects, and a bound selector $\beta$. The canonical working examples are $\mathbf{Pos}$ and $\mathbf{DCPO}$, with $\mathbf{Set}_{\mathrm{disc}}$ as a degenerate (order-trivial) reference case. Directed hypergraphs $\hyper$ are combinatorial structures generalizing graphs so that a hyperedge may connect arbitrary sets of vertices on each side; a sheaf on $\hyper$ is a functor from the incidence category $\Inc(\hyper)$ to $\datacat$. Throughout, we write $H^0(\hyper;\sheaf{L}) = \lim_{\Inc(\hyper)}\sheaf{L}$ for the \emph{global-section object} in $\datacat$, and $\Gamma(\hyper;\sheaf{L})=\Hom_\datacat(1,H^0(\hyper;\sheaf{L}))$ for the \emph{set of global sections}; the cohomological notation $H^0$ is intentional, anticipating that in coefficient categories carrying suitable additional structure, 
relative $H^0$ and $H^1$ will serve as obstruction classes in forthcoming work. For background on categories and limits see \cite{MacLane1998,Leinster2014}; on hypergraphs \cite{Bretto2013}; on network sheaves \cite{Curry2014,HansenGhrist2019}; on Tarski's theorem and complete lattices \cite{Tarski1955,DaveyPriestley2002}.
For recent discussion of sheaves on arbitrary finite posets and of the limitations of ad hoc hypergraph sheaf constructions, see \cite{AyzenbergGebhartMagaiSolomadin2025}.

\subsection{Summary of Contributions}
\label{subsec:summary}

We reformulate liability networks as sheaves on directed hypergraphs in which clearing configurations are global sections. Building on the lattice liability networks of \cite{GhristGouldLopezRiess2025}, which introduced section-like clearing conditions without sheaf structure, we provide a categorical framework that both encompasses and sharpens the previous formulation.

\textbf{Sheaf-Theoretic Clearing.} We construct a \style{liability sheaf} on a directed hypergraph that separates institution vertices (carrying payment states) from payment vertices (enforcing liability bounds), with distribution and collection hyperedges encoding the two sides of resource flow. Clearing configurations correspond precisely to global sections (Theorem~\ref{thm:clearing-sections}).

\textbf{Equalizer Presentation.} The global-section object is canonically the equalizer of the identity and a factored clearing operator $\Phi=A\circ D$ in the ambient liability category (Proposition~\ref{prop:limit-equalizer}). This identifies clearing as a finite-limit construction and makes the flow structure -- collective distribution $D$ followed by aggregation $A$ -- explicit at the object level, not merely at the level of fixed points.

\textbf{Clearing Invariance.} The construction is functorial under change of coefficient category: a finite-limit-preserving functor compatible with constraint subobjects induces a canonical isomorphism on global-section objects (Theorem~\ref{thm:clearing-invariance}). 

\textbf{Fixed-Point Consequences.} When payment objects have complete-lattice global elements, Tarski's theorem ensures that clearing sections exist and themselves form a complete lattice (Theorem~\ref{thm:lattice-existence}). When global elements carry a complete metric under which the clearing operator contracts, Banach's theorem yields a unique clearing section computable by iteration (Theorem~\ref{thm:banach-global}). Metric structure is imposed on global elements, not on the ambient liability category; this cleanly separates the categorical and analytic layers.

\textbf{Classical Recovery.} The Eisenberg-Noe model and lattice liability networks arise as special cases, obtained by choosing the coefficient category appropriately and specializing the distributor and aggregator morphisms (Sections~\ref{ssec:E-N} and~\ref{ssec:LLNs}).

\section{A Sheaf Model for Liabilities}
\label{sec:sheafmodel}

We begin by formalizing liability networks as decorated directed graphs, then show why this structure alone is insufficient for a sheaf-theoretic treatment. This motivates our construction of liability sheaves on directed hypergraphs, where clearing configurations emerge naturally as global sections.

\subsection{Liability Networks}
\label{ssec:liability-graphs}

A liability network begins with a directed graph whose vertices represent institutions and whose edges represent obligations. To model the diverse nature of resources and constraints in real-world networks, we decorate this graph with data from an appropriate category.

\begin{definition}[Liability Network]
\label{def:liability-graph}
Let $\datacat$ be a \style{liability category} in the sense of Definition~\ref{def:liability-category} below, with terminal object $1$ and bound selector $\beta$ assigning to each global element $\rho:1\to P$ a constraint subobject $\beta_P(\rho)\hookrightarrow P$. A \style{liability network} in $\datacat$ is a tuple $(G, \payob, \liab, \iota; \dis, \widehat\col)$ where $G = (V, E, s, t)$ is a directed graph with vertex set $V$, edge set $E$, and source/target maps $s, t: E \to V$. The decoration consists of:
\begin{itemize}
\item A collection $\payob = \{\payob_v \in \datacat : v \in V\}$ of \style{payment objects}
\item A collection $\liab = \{\liab_e: 1 \to \payob_{s(e)} : e \in E\}$ of \style{liability morphisms}
\item A collection $\iota = \{\iota_v: 1 \to \payob_v : v \in V\}$ of \style{exogenous resource morphisms}, selecting external resources at each institution
\item A collection $\dis = \{\dis_e: \payob_{s(e)} \to \beta_{\payob_{s(e)}}(\liab_e) : e \in E\}$ of \style{distributor morphisms}
\item A collection $\widehat\col = \{\widehat\col_v: \payob_v\times\prod_{t(e) = v} \beta_{\payob_{s(e)}}(\liab_e) \to \payob_v : v \in V\}$ of \style{aggregator morphisms}, taking exogenous resources and incoming payments as separate arguments
\end{itemize}
The \style{partial aggregator} used in the sheaf construction is the morphism
\[
\col_v\ :=\ \widehat\col_v\circ(\iota_v\times\id):\ \prod_{t(e)=v}\beta_{\payob_{s(e)}}(\liab_e)\ \longrightarrow\ \payob_v,
\]
obtained from $\widehat\col_v$ by partial evaluation along $\iota_v$ (using the canonical iso $\prod\beta\cong 1\times\prod\beta$): the exogenous resource is the structural point at which the binary aggregator $\widehat\col_v$ is curried into the unary morphism that appears in the sheaf.  When the data is clear from context, we abbreviate the liability network as $(G,\payob,\liab,\iota)$.  The forward reference to Definition~\ref{def:liability-category} is deliberate: the liability-category axioms isolate the minimal structure on $\datacat$ that makes the constraint subobjects $\beta_{\payob_{s(e)}}(\liab_e)$ canonically defined and the resulting clearing problem categorically well-posed.
\end{definition}

The payment object $\payob_v$ represents the space of possible payment states at vertex $v$, while the morphism $\liab_e$ encodes the nominal liability along edge $e$, and $\iota_v$ encodes the external resources available to vertex $v$. The morphisms from the terminal object provide a categorical mechanism for specifying particular elements or constraints within the payment objects. When $\datacat$ is concrete, these morphisms pick out specific values that represent liabilities and resources. This structure will be transformed into a sheaf on a hypergraph in Section~\ref{ssec:formal-construction}, where the liability morphisms will determine constraint subobjects at payment vertices.

The liability morphism $\liab_e: 1 \to \payob_{s(e)}$ specifies the obligation in the payment space of the debtor (source vertex). This modeling choice reflects that institutions define their obligations relative to their own capacity and payment structure.

\begin{remark}[Source-Typing and the Edge-Typed Generalization]
\label{rem:source-typing}
The convention that liabilities are valued in the debtor's payment object, $\liab_e:1\to\payob_{s(e)}$, with constraint subobject $\payob_{s(e)}^{\lambda_e}\subseteq\payob_{s(e)}$ a subobject of the source's payment object, suffices for single-currency scalar clearing networks -- Eisenberg--Noe \cite{EisenbergNoe2001} and the lattice liability networks of \cite{GhristGouldLopezRiess2025} -- in which every obligation is naturally typed by the debtor's payment capacity. An \style{edge-typed} generalization replaces $\payob_{s(e)}$ at the constraint level by an independent contract object $Q_e\in\datacat$, with liability $\liab_e:1\to Q_e$, constraint subobject $Q_e^{\lambda_e}=\beta_{Q_e}(\liab_e)$, distributor $\dis_e:\payob_{s(e)}\to Q_e^{\lambda_e}$, and aggregator data adapted accordingly; the present source-typed model is the special case $Q_e=\payob_{s(e)}$. The liability-sheaf construction, the equalizer characterization (Proposition~\ref{prop:limit-equalizer}), and the existence theorems extend to this setting with notational adjustments only. The edge-typed variant is the natural setting for multi-currency settlement, derivative contracts, and delivery-versus-payment, and is developed in forthcoming work; for the categorical and order-theoretic content of the present paper, source-typing suffices.
\end{remark}

\begin{example}[Classical Eisenberg-Noe as a Liability Network]
\label{ex:EN-liability-graph}
The classical Eisenberg-Noe model corresponds to a liability network where $\datacat$ is the category of partially ordered sets with monotone maps. Each payment object is $\payob_v = [0, \infty]$ with the usual ordering, each liability morphism $\liab_e: 1 \to [0, \infty]$ maps $\star\mapsto\ell_e\in[0,\infty]$, the scalar liability owed from institution $s(e)$ to institution $t(e)$, and each exogenous resource morphism $\iota_v: 1 \to [0, \infty]$ maps $\star\mapsto c_v\in[0,\infty]$, the external assets available to institution $v$. In this case, the terminal object $1$ is the one-element poset $\{\star\}$, and morphisms from $1$ to $[0,\infty]$ correspond to selecting specific values.
\end{example}

Despite its natural structure as data over a network, a liability network is not yet a sheaf. The payment objects $\payob_v$ sit at vertices without any compatibility conditions between them, and the liability morphisms $\liab_e$ decorate edges without specifying how resources actually flow. Most crucially, there is no notion of a ``section'' that would assign compatible payments throughout the network, and consequently no sheaf-theoretic machinery for analyzing existence and uniqueness of clearing.

To see why the graph structure alone is insufficient, observe that clearing requires solving a system of constraints: each institution must balance incoming and outgoing resources while respecting liability bounds. These constraints involve relationships between vertices that are mediated by edges, but in a liability network, edges are merely decorated with data rather than serving as spaces where compatibility is enforced. Furthermore, real networks involve one-to-many distributions (an institution paying multiple creditors) and many-to-one collections (an institution receiving from multiple debtors) that do not fit naturally into the structure of ordinary directed graphs.

These limitations motivate our next construction. By introducing payment vertices that mediate between institutions and organizing the resulting structure as a directed hypergraph, we can define a proper sheaf whose global sections correspond to clearing configurations. The hypergraph structure naturally accommodates the one-to-many and many-to-one relationships inherent in liability networks, while the sheaf formalism provides the mathematical machinery for analyzing existence and uniqueness of clearing.

\subsection{Hypergraph Sheaves}
\label{ssec:hyper}

The classical Eisenberg-Noe model operates on directed graphs where edges represent pairwise obligations. However, real liability networks involve fundamentally non-pairwise relationships: institutions distribute payments to multiple creditors simultaneously, and aggregate receipts from multiple debtors. This one-to-many and many-to-one structure necessitates a richer combinatorial framework. We develop a theory of sheaves on directed hypergraphs that captures these multi-way relationships while maintaining a clear categorical structure.

\begin{definition}[Directed Hypergraph]
\label{def:hypergraph}
A \style{directed hypergraph} is $\hyper = (V_\hyper, E_\hyper, \source{\cdot}, \target{\cdot})$ where $V_\hyper$ is a finite set of vertices, $E_\hyper$ is a finite set of hyperedges, and $\source{\cdot}: E_\hyper \to 2^{V_\hyper}$ and $\target{\cdot}: E_\hyper \to 2^{V_\hyper}$ are functions assigning to each hyperedge its sets of source and target vertices. We require that for each $\hypedg \in E_\hyper$, at least one of $\source{\hypedg}$ or $\target{\hypedg}$ is nonempty.
\end{definition}

A directed hyperedge $\hypedg$ with $\source{\hypedg} = \{v_1, \ldots, v_k\}$ and $\target{\hypedg} = \{w_1, \ldots, w_l\}$ represents a relationship where resources flow from the source vertices collectively to the target vertices. This generalizes directed edges (where $k = l = 1$) to capture multi-party transactions.

To define sheaves on directed hypergraphs, we must first encode the incidence structure categorically. The challenge is to respect both the hypergraph topology and the directionality of resource flows.

\begin{definition}[Incidence Category of a Directed Hypergraph]
\label{def:incidence-category}
Given a directed hypergraph $\hyper = (V_\hyper, E_\hyper, \source{\cdot}, \target{\cdot})$, its \style{incidence category} $\Inc(\hyper)$ has:
\begin{itemize}
\item Objects: $V_\hyper \sqcup E_\hyper$.
\item Morphisms: For each incidence $v \in \source{\hypedg}\cup \target{\hypedg}$, there is exactly one morphism $\hypedg \to v$, plus identity morphisms at each object. There are no non-trivial compositions: $\Inc(\hyper)$ has no morphisms of length greater than 1.
\end{itemize}
\end{definition}

Equivalently, $\Inc(\hyper)$ is the free category on the bipartite directed graph with arrows $E_\hyper\to V_\hyper$; as a category it is a bipartite quiver with no nontrivial compositions. This makes computing limits and colimits straightforward, as there are no commutativity conditions beyond those directly specified by the incidence morphisms.
Although $\Inc(\hyper)$ is combinatorially a bipartite quiver, we retain the hypergraph language because the source and target assignments $\source{\cdot}, \target{\cdot}$ carry information -- in particular, multi-element target sets -- that the incidence category alone does not record. For the liability hypergraph of Section~\ref{ssec:formal-construction}, the distinction between source and target vertices of a hyperedge determines whether a restriction map is a distributor or a projection, information that would require additional decoration in a pure quiver formulation.

\begin{definition}[Sheaf on a Directed Hypergraph]
\label{def:hypergraph-sheaf}
Let $\hyper$ be a directed hypergraph and $\datacat$ be a category. 
A \style{sheaf} on $\hyper$ with values in $\datacat$ is a (covariant) functor $\sheaf{F}: \Inc(\hyper) \to \datacat$.
\end{definition}

Concretely, a sheaf $\sheaf{F}$ assigns:
\begin{itemize}
\item To each vertex $v \in V_\hyper$: an object $\stalk{\sheaf{F}}{v} \in \datacat$ (the \style{stalk at $v$})
\item To each hyperedge $\hypedg \in E_\hyper$: an object $\stalk{\sheaf{F}}{\hypedg} \in \datacat$ (the \style{stalk at $\hypedg$})
\item To each morphism $\hypedg \to v$ in $\Inc(\hyper)$: a morphism $\sheaf{F}(\hypedg \to v): \stalk{\sheaf{F}}{\hypedg} \to \stalk{\sheaf{F}}{v}$ in $\datacat$ (the \style{restriction map})
\end{itemize}

\begin{remark}
We define sheaves as covariant functors on $\Inc(\hyper)$ where arrows go from hyperedges to vertices. This is equivalent to the standard approach using contravariant functors on the opposite category, but aligns more naturally with our economic interpretation where resources flow from distribution/collection hyperedges to institution vertices.
\end{remark}

\begin{definition}[Global Sections]
\label{def:global-sections}
Let $\sheaf{F}$ be a sheaf on hypergraph $\hyper$ with values in a category $\datacat$ with finite limits. The \style{global-section object} is the limit
\[
H^0(\hyper; \sheaf{F})\ :=\ \lim_{\Inc(\hyper)} \sheaf{F}\ \in\ \datacat,
\]
and the \style{set of global sections} is
\[
\Gamma(\hyper; \sheaf{F})\ :=\ \Hom_{\datacat}\!\bigl(1,\, H^0(\hyper;\sheaf{F})\bigr).
\]
When $\datacat$ is concrete, we identify $\Gamma(\hyper;\sheaf{F})$ with the set of compatible section data assembled from global elements of the stalks.
\end{definition}

The notation $H^0$ is standard for the zeroth cohomology of a cellular sheaf; the present paper uses only $H^0$ and does not develop higher cohomology. Throughout, \emph{$H^0$ denotes the limit object in $\datacat$ and $\Gamma$ denotes its set of global elements}; statements about clearing as a categorical construction are phrased in terms of $H^0$, while statements about clearing configurations as concrete payment data are phrased in terms of $\Gamma$.

\begin{remark}
A global section $\sigma \in \Gamma(\hyper; \sheaf{F})$ corresponds to a collection of elements $\{\sigma_x \in \stalk{\sheaf{F}}{x}\}_{x \in V_\hyper \sqcup E_\hyper}$ (via the universal property of limits) such that for each morphism $\phi: x \to y$ in $\Inc(\hyper)$, we have $\sheaf{F}(\phi)(\sigma_x) = \sigma_y$. When working in concrete categories like $\mathbf{Pos}$ or $\mathbf{Set}$, we freely use this element-wise description.
\end{remark}

Since $\Inc(\hyper)$ is finite when $\hyper$ is finite, the limit defining global sections exists whenever $\datacat$ has finite products and equalizers. This is a weaker requirement than completeness and is satisfied by all our example categories.

\begin{example}[Distribution Hypergraph]
Consider a vertex $v$ that must distribute resources to vertices $w_1, w_2, w_3$. We model this with a hyperedge $\hypedg$ where $\source{\hypedg} = \{v\}$ and $\target{\hypedg} = \{w_1, w_2, w_3\}$. A sheaf on this structure encodes distribution through the stalk $\stalk{\sheaf{F}}{\hypedg}$ and enforces compatibility via morphisms $\stalk{\sheaf{F}}{\hypedg} \to \stalk{\sheaf{F}}{v}$ and $\stalk{\sheaf{F}}{\hypedg} \to \stalk{\sheaf{F}}{w_i}$.
\end{example}

The hypergraph sheaf framework provides the mathematical foundation for modeling complex liability networks. By separating vertices (representing institutions) from hyperedges (representing transactions), we can encode arbitrary distribution and aggregation patterns while maintaining a clean categorical structure. In the next subsections, we specialize this framework to construct liability sheaves that model financial clearing.

\subsection{Categorical Liabilities}
\label{ssec:categorical}

Consider a liability network $(G, \payob, \liab, \iota)$ from Definition \ref{def:liability-graph}. Intuitively, $\payob_v$ represents the set of possible payment states at institution $v$, $\liab_e$ specifies the ``maximum'' obligation from institution $s(e)$ to institution $t(e)$ along edge $e$, and $\iota_v$ specifies external resources available to institution $v$. All liabilities, resources, and payments are objects in a chosen data category. To make sense of payment flows and to facilitate defining clearing sections, we impose structural requirements on the category in which payments are valued. These requirements ensure that liability constraints can be represented categorically and that the resulting sheaf has well-behaved global sections.

\begin{definition}[Liability Category]
\label{def:liability-category}
A category $\datacat$ is a \style{liability category} if:

\begin{enumerate}
\item \textbf{Finite Limits:} $\datacat$ has a terminal object $1$, all finite products, and equalizers.

\item \textbf{Ordered Global Elements:} For each object $P$, the set $\Hom(1,P)$ is equipped with a partial order $\leqslant$, and for any morphism $f:P\to Q$ the induced map $f_\ast:\Hom(1,P)\to\Hom(1,Q)$ is order-preserving.

\item \textbf{Constraint Subobjects:} For each $P$, a distinguished class $\mathcal S_P$ of monomorphisms $m: P^c\hookrightarrow P$ (one chosen representative per subobject class), called \style{constraint subobjects}.  We require stability under pullback: if $m: P^c\hookrightarrow P$ lies in $\mathcal S_P$ and $f:Q\to P$, then a chosen pullback $f^\ast m:Q^c\hookrightarrow Q$ lies in $\mathcal S_Q$.  By abuse of notation, we identify each element of $\mathcal{S}_P$ with its domain object $P^c$ when no confusion arises.

\item \textbf{Bound Selector:} An assignment $\beta_P: \Hom(1,P) \to \mathcal{S}_P$ 
for each object $P$, associating to each global element a constraint subobject.

\item \textbf{Product-Order Compatibility:} For every finite family $(P_i)_{i\in I}$ of objects, the canonical map
\[
\Hom\!\bigl(1,\textstyle\prod_{i\in I} P_i\bigr)\ \longrightarrow\ \prod_{i\in I}\Hom(1,P_i)
\]
induced by the product projections is an order isomorphism, where the codomain carries the componentwise (product) order.  Componentwise completeness of the product when each factor is a complete lattice is then standard order theory and is invoked, where needed, in the proof of Theorem~\ref{thm:lattice-existence}.
\end{enumerate}
\end{definition}

\begin{remark}[Pullback Stability of Constraint Classes]
\label{rem:pullback-stability}
The pullback-stability requirement in axiom (3) constrains the
choice of $\mathcal{S}_P$.  In $\mathbf{Pos}$, principal
down-sets $\downarrow\!p$ are \emph{not} generally stable under
pullback by monotone maps: if $f:Q\to P$ is monotone, the
preimage $f^{-1}(\downarrow\!p)=\{q\in Q:f(q)\leqslant p\}$ is a
down-set but need not be principal.  Two remedies are available:
(i)~enlarge $\mathcal{S}_P$ to the class of all down-sets, which
is stable under pullback by any monotone map; or (ii)~restrict
morphisms to residuated maps (those possessing right adjoints),
under which principal down-sets pull back to principal down-sets.
Throughout this paper we adopt option~(i) for $\mathbf{Pos}$:
the constraint class consists of all down-sets, with principal
down-sets arising as the images of the bound selector $\beta$.
In $\mathbf{DCPO}$, the natural constraint class consists of Scott-closed downsets, and Scott-closed downsets are pullback-stable under Scott-continuous maps; the bound selector $\beta_P(\rho)=\downarrow\!\rho(\star)$ already lands in this class, since principal ideals in a dcpo are Scott-closed.  The concrete choices for each category are recorded in Remark~\ref{rem:constraint-practice} below.
\end{remark}

\begin{remark}
For a liability morphism $\liab: 1 \to P$ with $\lambda = \liab(\star)$, we write $P^{\lambda} := \beta_P(\liab)$ for the corresponding constraint subobject.
\end{remark}

\begin{remark}[Constraint Subobjects in Practice]
\label{rem:constraint-practice}
In concrete settings, constraint subobjects model bounds or feasibility domains:
\begin{itemize}
\item In $\mathbf{Pos}$: all down-sets $D = \{x \in P : x \leqslant d \text{ for some } d \in D\}$, which are stable under pullback by monotone maps. Principal down-sets $\downarrow\!\ell=\{x:x\leqslant \ell\}$ form an important subclass but require special care under general morphisms;
\item In $\mathbf{DCPO}$ (directed-complete partial orders with Scott-continuous maps): Scott-closed downsets, which are stable since Scott-continuous preimages of Scott-closed sets are Scott-closed; principal ideals $\downarrow\!\ell$ form an important Scott-closed subclass and serve as the image of the bound selector.
\end{itemize}
\end{remark}

\begin{example}[Payment Bounds]
\label{ex:bounded-subobject}
In $\mathbf{Pos}$, for a liability morphism $\liab: 1 \to P$ with liability datum $\lambda = \liab(\star)$, the constraint subobject $P^{\lambda}$ is the down-set $\{s \in P : s \leqslant \lambda\}$ with the inclusion map. In $\mathbf{DCPO}$, for $\liab: 1 \to L$ with $\lambda = \liab(\star)$, the constraint subobject is the principal ideal $\downarrow\!\lambda = \{x \in L : x \leqslant \lambda\}$, which is Scott-closed and equipped with the induced dcpo structure; the inclusion is Scott-continuous.
\end{example}

Certain categories will play prominent roles in examples of later sections. For completeness, we briefly verify that these are liability categories.

\begin{lemma}[Examples of Liability Categories]
\label{lem:liability-categories}
The following are liability categories in the sense of Definition \ref{def:liability-category}:
\begin{enumerate}
\item $\mathbf{Set}_{\mathrm{disc}}$: Sets with functions, where $\Hom(1, X) = X$ is equipped with the discrete order (all elements incomparable except to themselves), and constraint subobjects are subsets. The bound selector $\beta_X(x) = \{x\} \subseteq X$ maps each element to its singleton subset.
\item $\mathbf{Pos}$: Partially ordered sets with monotone maps, where $\Hom(1, P) = P$ uses the given order, and constraint subobjects are all down-sets (subsets $D \subseteq P$ where $x \in D$ and $y \leqslant x$ implies $y \in D$). The bound selector $\beta_P(p) = \downarrow\!p(\star) = \{x \in P : x \leqslant p(\star)\}$ maps to principal down-sets, a special case of down-sets.
\item $\mathbf{DCPO}$: Directed-complete partial orders with Scott-continuous maps, where $\Hom(1, P) = P$ inherits the dcpo order (the terminal object is the one-point dcpo), and constraint subobjects are Scott-closed downsets, with bound selector $\beta_P(\rho) = \downarrow\!\rho(\star) = \{x \in P : x \leqslant \rho(\star)\}$.
\end{enumerate}
\end{lemma}

{\em Proof:} We verify the conditions of Definition \ref{def:liability-category} for each category.

\textbf{(1) Finite Limits.} All three categories have terminal objects: $\{\star\}$ in $\mathbf{Set}_{\mathrm{disc}}$, the one-element poset in $\mathbf{Pos}$, and the one-point dcpo in $\mathbf{DCPO}$. Finite products are Cartesian products with componentwise order in $\mathbf{Pos}$ and $\mathbf{DCPO}$, and Cartesian products of sets in $\mathbf{Set}_{\mathrm{disc}}$. Equalizers are subsets where functions agree in $\mathbf{Set}_{\mathrm{disc}}$, sub-posets with induced order in $\mathbf{Pos}$, and sub-dcpos with induced order in $\mathbf{DCPO}$. Closure of the equalizer under directed suprema follows from Scott-continuity: if $f,g:P\to Q$ are Scott-continuous and $D\subseteq\{p:f(p)=g(p)\}$ is directed, then $f(\sup D)=\sup f(D)=\sup g(D)=g(\sup D)$.

\textbf{(2) Ordered Global Elements.} The partial orders on $\Hom(1, -)$ are specified above.

\textbf{(3) Constraint Subobjects.} Each category has natural classes of constraint subobjects that are stable under pullback:
\begin{itemize}
\item In $\mathbf{Set}_{\mathrm{disc}}$: Any subset $S \subseteq X$ with inclusion; the bound selector produces singleton sets $\beta_X(x) = \{x\}$
\item In $\mathbf{Pos}$: All down-sets, with the bound selector producing principal down-sets $\downarrow\!p = \{x \in P : x \leqslant p\}$ as a special case
\item In $\mathbf{DCPO}$: Scott-closed downsets, with the bound selector producing principal ideals $\downarrow\!\rho = \{x \in P : x \leqslant \rho\}$ as a Scott-closed special case
\end{itemize}

\textbf{(4) Bound Selector.} The assignments $\beta$ are specified above: $\beta_X(x) = \{x\}$ in $\mathbf{Set}_{\mathrm{disc}}$; $\beta_P(p) = \downarrow\!p(\star)$ in $\mathbf{Pos}$; $\beta_P(\rho) = \downarrow\!\rho(\star)$ in $\mathbf{DCPO}$. In each case, $\beta_P(\rho) \in \mathcal{S}_P$ for every global element $\rho$. The order-preservation requirement in axiom~(2) is immediate: $\mathbf{Set}_{\mathrm{disc}}$ has discrete orders; morphisms in $\mathbf{Pos}$ are monotone by definition; Scott-continuous maps in $\mathbf{DCPO}$ are in particular monotone.

\textbf{(5) Product-Order Compatibility.} In each of the three categories the underlying-set functor preserves finite products, so the canonical map $\Hom(1,\prod_i P_i)\to\prod_i\Hom(1,P_i)$ is a bijection. The product order on $\Hom(1,\prod_i P_i)$ is by construction componentwise in $\mathbf{Pos}$ and $\mathbf{DCPO}$ (where finite products carry the componentwise order), and trivially componentwise in $\mathbf{Set}_{\mathrm{disc}}$ (where every order is discrete); the canonical map is therefore an order isomorphism.
\qed

\begin{remark}[Role of $\mathbf{DCPO}$]
\label{rem:dcpo-role}
The two main working categories play complementary roles: $\mathbf{Pos}$ supplies the broadest setting in which Tarski's theorem yields existence of clearing sections (Theorem~\ref{thm:lattice-existence}), while $\mathbf{DCPO}$ refines this to the Scott-continuous regime in which Kleene iteration converges to the least clearing section (Subsection~\ref{ssec:kleene}). Passage between the two via the forgetful functor $U:\mathbf{DCPO}\to\mathbf{Pos}$ is captured by Clearing Invariance (Example~\ref{ex:dcpo-to-pos}).
\end{remark}

\begin{remark}
Categories of ordered abelian groups ($\mathbf{oAb}$) and ordered vector spaces are natural candidates for liability categories. However, in these categories the terminal object is the zero object, so $\Hom(1, P) = \{0\}$ is a singleton and does not recover arbitrary elements of $P$. To obtain a useful supply of global elements, one must pass to a pointed or affine variant  --  for instance, by replacing the terminal object $1$ with a designated ``base'' object, or by equipping each $P$ with an underlying-set functor providing the order structure required by axiom~(2).
\end{remark}

\subsection{Liability Sheaves}
\label{ssec:formal-construction}

We now specialize the hypergraph sheaf framework to model liability networks with resource flows, obligations, and clearing constraints. Starting from a liability network representing pairwise obligations between institutions, we systematically construct a hypergraph that separates payment flows from institutional states, then define a sheaf whose global sections correspond to clearing configurations.

\begin{definition}[Liability Hypergraph Construction]
\label{def:liability-hypergraph}
Given a liability network $(G, \payob, \liab, \iota)$, we construct the \style{liability hypergraph} $\hyper_G = (V_{\hyper}, E_{\hyper}, \source{\cdot}, \target{\cdot})$ as follows.

The vertex set $V_{\hyper} = V \sqcup V_{\mathrm{pay}}$ consists of the original \style{institution vertices} $V$ and \style{payment vertices} $V_{\mathrm{pay}} = \{e^* : e \in E\}$, one for each edge in $G$.

The hyperedge set decomposes as $E_{\hyper} = E_{\mathrm{dis}} \sqcup E_{\mathrm{col}}$:
\begin{itemize}
\item \textbf{Distribution hyperedges:} for each $v \in V$, introduce $h_v^{\dis}$ with 
\[
\source{h_v^{\dis}} = \{v\},\qquad \target{h_v^{\dis}} = \{e^* : e \in E,\ s(e) = v\}.
\]
\item \textbf{Collection hyperedges:} for each $v \in V$, introduce $h_v^{\col}$ with 
\[
\source{h_v^{\col}} = \emptyset,\qquad \target{h_v^{\col}} = \{e^* : e \in E,\ t(e) = v\}\ \cup\ \{v\}.
\]
\end{itemize}
When $t^{-1}(v) = \emptyset$, the product $\prod_{t(e)=v}\payob_{s(e)}^{\lambda_e}$ is the empty product, interpreted as the terminal object $1$ in $\datacat$. In this case, $\col_v: 1 \to \payob_v$ is still defined according to the specific model's aggregation rule, where the contribution from incoming payments is typically an appropriate identity element (e.g., $0$ for summation, $\delta_0$ for convolution).

Thus, in the incidence category $\Inc(\hyper_G)$ the generating arrows are
\[
h_v^{\dis}\to v,\quad h_v^{\dis}\to e^*\ \ (s(e)=v);\qquad
h_v^{\col}\to e^*\ \ (t(e)=v),\quad h_v^{\col}\to v.
\]
\end{definition}

\begin{remark}
The choice $\source{h_v^{\col}}=\emptyset$ is consistent with Definition~\ref{def:hypergraph}, since $\target{h_v^{\col}}=\{e^* : t(e)=v\}\cup\{v\}$ is nonempty.  Setting the source empty ensures that $h_v^{\col}$ \style{projects} to each incoming payment vertex and to $v$, aligning the categorical universal property of products with our use of a product stalk at $h_v^{\col}$ and removing any need for non-canonical maps $e^*\to h_v^{\col}$.
\end{remark}

Each institution $v$ has exactly one outgoing distribution hyperedge $h_v^{\dis}$ that connects to all payment vertices for its outgoing obligations, and exactly one incoming collection hyperedge $h_v^{\col}$ that aggregates all incoming payments. This bipartite structure cleanly separates the distribution and collection operations.

We can now define the central object of our framework.

\begin{definition}[Liability Sheaf]
\label{def:liability-sheaf}
Let $(G,\payob,\liab,\iota)$ be valued in a liability category $\datacat$ (Definition~\ref{def:liability-category}). The \style{liability sheaf} is a functor $\sheaf{L}: \Inc(\hyper_G) \to \datacat$ defined by:

\medskip
\noindent\textbf{Stalks.}
\[
\stalk{\sheaf{L}}{v}=\payob_v\quad (v\in V),\qquad
\stalk{\sheaf{L}}{e^*}=\beta_{\payob_{s(e)}}(\liab_e)\quad (e^*\in V_{\mathrm{pay}}),
\]
\[
\stalk{\sheaf{L}}{h_v^{\dis}}=\payob_v,\qquad
\stalk{\sheaf{L}}{h_v^{\col}}=\prod_{t(e)=v}\payob_{s(e)}^{\lambda_e}.
\]

\noindent\textbf{Restriction Maps.}
For morphisms in $\Inc(\hyper_G)$, the sheaf assigns:
\[
\sheaf{L}(h_v^{\dis} \to v)=\id_{\payob_v},\qquad
\sheaf{L}(h_v^{\dis} \to e^*)=\dis_e:\payob_v\to \payob_{s(e)}^{\lambda_e}\quad (s(e)=v),
\]
\[
\sheaf{L}(h_v^{\col} \to e^*)=\pi_e:\ \prod_{t(e')=v}\payob_{s(e')}^{\lambda_{e'}}\ \longrightarrow\ \payob_{s(e)}^{\lambda_e}\quad (t(e)=v),
\]
\[
\sheaf{L}(h_v^{\col} \to v)=\col_v:\ \prod_{t(e)=v}\payob_{s(e)}^{\lambda_e}\ \longrightarrow\ \payob_v,
\]
\end{definition}

Throughout, we write $\payob_{s(e)}^{\lambda_e} := \beta_{\payob_{s(e)}}(\liab_e)$ where $\lambda_e = \liab_e(\star)$ is the liability datum selected by the morphism $\liab_e: 1 \to \payob_{s(e)}$.  The partial aggregators $\col_v$ used in the sheaf are derived from the data $(\widehat\col_v,\iota_v)$ as in Definition~\ref{def:liability-graph}.

\begin{remark}[Empty-Product Convention]
\label{rem:exogenous}
When $t^{-1}(v)=\emptyset$, the empty product is the terminal object $1$, and the partial aggregator reduces to $\col_v=\widehat\col_v\circ(\iota_v\times\id_1):1\to\payob_v$; on global elements, $\col_v(\star)=\widehat\col_v(\iota_v(\star),\star)$.  In concrete categories, this is the value of $\widehat\col_v$ at the exogenous resource alone with the appropriate identity element for the aggregation operation (e.g.\ $0$ for summation, $\delta_0$ for convolution).
\end{remark}

The principal motivation for adopting this sheaf formalism is the ability to frame the clearing problem as a problem of global sections of the sheaf.

\begin{theorem}[Clearing Sections Characterization]
\label{thm:clearing-sections}
Let $\sheaf{L}$ be the liability sheaf on $\hyper_G$. A global section $\sigma \in \Gamma(\hyper_G;\sheaf{L})$ is uniquely determined by elements
\[
x_v\in\payob_v\quad (v\in V),\qquad p_e\in\payob_{s(e)}^{\lambda_e}\quad (e\in E),
\]
satisfying
\[
\boxed{\ p_e=\dis_e\big(x_{s(e)}\big)\ \ \text{ for all }e\in E\ }\qquad\text{and}\qquad
\boxed{\ x_v=\col_v\big((p_e)_{t(e)=v}\big)\ \ \text{ for all }v\in V\ }.
\]
\end{theorem}

{\em Proof:} A global section $\sigma \in \Gamma(\hyper_G; \sheaf{L})$ corresponds to a collection of elements:
\begin{itemize}
\item $\sigma_v \in \stalk{\sheaf{L}}{v} = \payob_v$ for each institution vertex $v \in V$
\item $\sigma_{e^*} \in \stalk{\sheaf{L}}{e^*} = \payob_{s(e)}^{\lambda_e}$ for each payment vertex $e^* \in V_{\mathrm{pay}}$
\item $\sigma_h \in \stalk{\sheaf{L}}{h}$ for each hyperedge $h \in E_{\hyper}$
\end{itemize}
These must satisfy compatibility conditions from the sheaf structure.

For the distribution hyperedge $h_v^{\dis}$, compatibility requires $\sheaf{L}(h_v^{\dis} \to v)(\sigma_{h_v^{\dis}}) = \sigma_v$. Since $\sheaf{L}(h_v^{\dis} \to v) = \id_{\payob_v}$, we must have $\sigma_{h_v^{\dis}} = \sigma_v$. Let us denote this common element by $x_v \in \payob_v$.

Next, for each edge $e$ with $s(e) = v$, compatibility requires $\sheaf{L}(h_v^{\dis} \to e^*)(\sigma_{h_v^{\dis}}) = \sigma_{e^*}$. Since $\sheaf{L}(h_v^{\dis} \to e^*) = \dis_e$ and $\sigma_{h_v^{\dis}} = x_v$, we have $\sigma_{e^*} = \dis_e(x_v)$. Let us denote $\sigma_{e^*}$ by $p_e$, giving us the first boxed equation: $p_e = \dis_e(x_{s(e)})$.

For the collection hyperedge $h_v^{\col}$, the restrictions $\sheaf{L}(h_v^{\col} \to e^*)$ are projections $\pi_e$ from the product. Thus $\sigma_{h_v^{\col}} \in \prod_{t(e)=v} \payob_{s(e)}^{\lambda_e}$ must satisfy $\pi_e(\sigma_{h_v^{\col}}) = \sigma_{e^*} = p_e$ for each incoming edge. This forces $\sigma_{h_v^{\col}} = (p_e)_{t(e)=v}$.

Finally, compatibility requires $\sheaf{L}(h_v^{\col} \to v)(\sigma_{h_v^{\col}}) = \sigma_v$. Since $\sheaf{L}(h_v^{\col} \to v) = \col_v$ and $\sigma_{h_v^{\col}} = (p_e)_{t(e)=v}$, we obtain $x_v = \col_v((p_e)_{t(e)=v})$, yielding the second boxed equation.

Conversely, given any collection $(x_v)_{v \in V}$ and $(p_e)_{e \in E}$ satisfying both boxed equations, we can define $\sigma$ by setting $\sigma_v = x_v$, $\sigma_{e^*} = p_e$, $\sigma_{h_v^{\dis}} = x_v$, and $\sigma_{h_v^{\col}} = (p_e)_{t(e) = v}$. The equations ensure all compatibility conditions are satisfied, yielding a global section. Thus $\sigma$ is uniquely determined by the pair $(\{x_v\}_{v \in V}, \{p_e\}_{e \in E})$ satisfying the two boxed equations.
\qed

\begin{corollary}
The clearing problem for a liability network reduces to computing the set of global sections $\Gamma(\hyper_G; \sheaf{L})$ of the associated liability sheaf.
\end{corollary}

Theorem~\ref{thm:clearing-sections} characterizes global sections
element-wise.  The following proposition upgrades this to an
object-level identification in $\datacat$, expressing the
global-section object as an equalizer of the clearing operator.

\begin{proposition}[Limit as Equalizer]
\label{prop:limit-equalizer}
Let\/ $\sheaf{L}$ be a liability sheaf on\/ $\hyper_G$ valued in a
liability category $\datacat$.  Set
\[
P \;=\; \prod_{v\in V}\payob_v ,\qquad
B \;=\; \prod_{e\in E}\payob_{s(e)}^{\lambda_e} ,
\]
and define the \style{collective distribution}
$D:P\to B$ and \style{collective aggregation} $A:B\to P$ by
\[
D(\mathbf{x})_e = \dis_e(x_{s(e)}),\qquad
A(\mathbf{p})_v = \col_v\!\bigl((p_e)_{t(e)=v}\bigr).
\]
Then the clearing operator is $\Phi=A\circ D:P\to P$, and
there is a canonical isomorphism in $\datacat$:
\[
\lim_{\Inc(\hyper_G)}\sheaf{L}
\;\cong\;
\Eq(\id_P,\,\Phi).
\]
\end{proposition}

{\em Proof:}
We verify that $\Eq(\id_P,\Phi)$ satisfies the
universal property of $\lim\sheaf{L}$.  A cone over
$\sheaf{L}$ from an object $T$ consists of morphisms to every
stalk, compatible with every restriction map.  Among the four
families of restriction maps in $\sheaf{L}$
(Definition~\ref{def:liability-sheaf}), three impose no
independent conditions:
the identity restriction
$\sheaf{L}(h_v^{\dis}\!\to v)=\id_{\payob_v}$ forces the
distribution-hyperedge component to equal the institution-vertex
component $f_v:T\to\payob_v$;
the distributor restrictions
$\sheaf{L}(h_v^{\dis}\!\to e^*)=\dis_e$ force the
payment-vertex component to be $\dis_e\circ f_{s(e)}$;
the projection restrictions
$\sheaf{L}(h_v^{\col}\!\to e^*)=\pi_e$ force the
collection-hyperedge component to be
$(\dis_e\circ f_{s(e)})_{t(e)=v}$.
Thus the entire cone is determined by a single morphism
$f=(f_v):T\to P$, subject to the remaining condition from
$\sheaf{L}(h_v^{\col}\!\to v)=\col_v$:
\[
\col_v\!\bigl((\dis_e\circ f_{s(e)})_{t(e)=v}\bigr) = f_v
\qquad\text{for all }v\in V,
\]
which is $\Phi\circ f = f$.
A morphism $f:T\to P$ with $\Phi\circ f=f$ corresponds
bijectively and naturally in $T$ to a morphism
$T\to\Eq(\id_P,\Phi)$.
\qed

\begin{remark}[Three Formulations of Clearing]
\label{rem:three-formulations}
Proposition~\ref{prop:limit-equalizer} identifies three
equivalent formulations of the clearing problem, at
progressively finer levels of structure:
\begin{itemize}
\item \style{Operator-theoretic:}  Clearing states are
fixed points of $\Phi:P\to P$.  Tarski's or Banach's theorem
applies directly at this level.
\item \style{Sheaf-theoretic:}  The global-section object
$H^0(\hyper_G;\sheaf{L})=\lim\sheaf{L}\cong\Eq(\id_P,\Phi)$ lives in
$\datacat$; its set of global elements recovers the fixed points via
\[
\Gamma(\hyper_G;\sheaf{L})=\Hom_\datacat(1,H^0(\hyper_G;\sheaf{L}))\cong\Fix(\Phi_\ast).
\]
\item \style{Structured:}  The factorization $\Phi=A\circ D$
decomposes the clearing operator into local
operations -- distribution $D$ and aggregation $A$ ---
governed by the hypergraph topology.  This constrains the
class of admissible operators and makes the flow structure
explicit.
\end{itemize}
The first level suffices for existence and uniqueness; the
second provides an object in $\datacat$ on which functorial
constructions (e.g.\ morphisms between networks) can act;
the third is the structural content of the sheaf formalism.
\end{remark}

\begin{proposition}[Institution-Edge Duality]
\label{prop:edge-side-equalizer}
With notation as in Proposition~\ref{prop:limit-equalizer}, define the \style{edge-side clearing operator}
$\Phidual = D\circ A:B\to B$.
The collective distribution and aggregation restrict to mutually inverse isomorphisms
\[
D:\ \Eq(\id_P,\,\Phi)\ \xrightarrow{\ \cong\ }\ \Eq(\id_B,\,\Phidual),
\qquad
A:\ \Eq(\id_B,\,\Phidual)\ \xrightarrow{\ \cong\ }\ \Eq(\id_P,\,\Phi)
\]
in $\datacat$.  In particular, $H^0(\hyper_G;\sheaf{L})\cong\Eq(\id_B,\,D\circ A)$, and on global elements the induced maps $D_\ast:\Fix(\Phi_\ast)\to\Fix(\Phidualast)$ and $A_\ast:\Fix(\Phidualast)\to\Fix(\Phi_\ast)$ are mutually inverse bijections.
\end{proposition}

{\em Proof:}
By the universal property of the equalizer, a morphism $f:T\to P$ factors through $\Eq(\id_P,\Phi)$ iff $\Phi\circ f=f$, equivalently $A\circ D\circ f=f$.  Postcomposing with $D$ gives $D\circ A\circ D\circ f=D\circ f$, exhibiting $D\circ f:T\to B$ as factoring through $\Eq(\id_B,\Phidual)$; this defines a morphism $\Eq(\id_P,\Phi)\to\Eq(\id_B,\Phidual)$, by abuse of notation again denoted $D$.  Symmetrically, $A:B\to P$ induces a morphism $\Eq(\id_B,\Phidual)\to\Eq(\id_P,\Phi)$.  The composites are identities on the respective equalizers: on $\Eq(\id_P,\Phi)$ the composite sends $f$ to $A\circ D\circ f=\Phi\circ f=f$, and on $\Eq(\id_B,\Phidual)$ it sends $g$ to $D\circ A\circ g=\Phidual\circ g=g$.  Hence the two restrictions are mutually inverse isomorphisms.  Applying $\Hom_\datacat(1,-)$ gives the bijection on fixed-point sets.
\qed

\begin{remark}[Content Beyond the Universal Property]
\label{rem:edge-side-content}
Proposition~\ref{prop:edge-side-equalizer} depends on the factorization $\Phi=A\circ D$ supplied by the sheaf, not on $\Phi$ alone, and is therefore invisible to a black-box treatment of the clearing operator.  Two consequences are practical:
\begin{itemize}
\item \style{Edge-side fixed-point arguments.}  When the edge-side $B=\prod_e\payob_{s(e)}^{\lambda_e}$ is order-theoretically or metrically more tractable than the institution-side $P=\prod_v\payob_v$ -- for instance when constraint subobjects are bounded while payment objects are not, as in Eisenberg--Noe with $\payob_v=[0,\infty]$ and $\payob_{s(e)}^{\lambda_e}=[0,\ell_e]$ -- existence and uniqueness arguments may be run on $\Phidual$ and transported to $\Phi$ via $A$.
\item \style{Iteration on either side.}  When the edge-side object $B$ carries the order-theoretic or metric structure required by Theorem~\ref{thm:kleene} or Theorem~\ref{thm:banach-global}, and $\Phidual=D\circ A$ satisfies the corresponding monotonicity, Scott-continuity, or contraction hypothesis, the fixed-point argument may be run on $\Phidual$ and transported to $\Phi$ via $A$.  The choice of side is dictated by which operator satisfies the useful hypotheses in a given example.
\end{itemize}
\end{remark}

\subsection{Functoriality and Clearing Invariance}
\label{ssec:functoriality}

The global-section object $H^0(\hyper_G;\sheaf{L})=\lim_{\Inc(\hyper_G)}\sheaf{L}$ is a finite limit in $\datacat$, and limits are natural in two directions: in the sheaf (fixing the base hypergraph) and in the ambient category (fixing the diagram). The first is a one-paragraph lemma; the second is the \style{Clearing Invariance Theorem}, which asserts that clearing is preserved under any strict, finite-limit-preserving change of coefficient category. Together they say that clearing is a \style{natural} construction.

\begin{lemma}[Fixed-Base Functoriality]
\label{lem:fixed-base-functoriality}
Let $\sheaf{L}, \sheaf{M}: \Inc(\hyper_G) \to \datacat$ be liability sheaves over the same liability hypergraph. Every natural transformation $\alpha:\sheaf{L}\Rightarrow\sheaf{M}$ induces a canonical morphism
\[
H^0(\alpha):\ H^0(\hyper_G;\sheaf{L})\ \longrightarrow\ H^0(\hyper_G;\sheaf{M})
\]
in $\datacat$, and the assignment $\alpha\mapsto H^0(\alpha)$ is functorial in $\alpha$.
\end{lemma}

{\em Proof:} The global-section object is the limit $\lim_{\Inc(\hyper_G)}\sheaf{L}$ (Proposition \ref{prop:limit-equalizer} identifies it with an equalizer). The $\lim$-functor on the diagram category $\datacat^{\Inc(\hyper_G)}$ sends a natural transformation to the universal morphism between limit cones; functoriality is the universal property. \qed

We now turn to the cross-category direction. Let $(\datacat,\mathcal S,\beta)$ and $(\datacat',\mathcal S',\beta')$ be liability categories.

\begin{definition}[Morphism of Liability Categories]
\label{def:morphism-liability-category}
A \style{morphism of liability categories} $F:(\datacat,\mathcal S,\beta)\to(\datacat',\mathcal S',\beta')$ consists of:
\begin{enumerate}
\item A functor $F:\datacat\to\datacat'$ that preserves finite limits.
\item For each object $\payob\in\datacat$ and each global element $\rho:1\to\payob$, a monomorphism
\[
\phi_{\payob,\rho}:\ F\bigl(\beta_{\payob}(\rho)\bigr)\ \hookrightarrow\ \beta'_{F\payob}\!\bigl(F\rho\bigr)
\]
in $\datacat'/F\payob$, where $F\rho:1'\to F\payob$ is identified via the canonical isomorphism $1'\cong F1$ provided by terminal-object preservation. We call $\phi$ the \style{compatibility comparison} for $F$.
\end{enumerate}
The morphism is \style{strict} if every $\phi_{\payob,\rho}$ is an isomorphism, and \style{lax} otherwise.
\end{definition}

\begin{remark}[Role of the Compatibility Comparison]
\label{rem:phi-role}
The compatibility comparison $\phi$ is what lets a functor between ambient categories interact coherently with the constraint structure that defines payment bounds. Strict morphisms are the typical case: full-subcategory inclusions $\datacat'\hookrightarrow\datacat$ respecting the constraint classes, and forgetful functors for which the target constraint class is pulled back from the source. Laxness arises when the target category admits a strictly richer notion of constraint than the source, so that $\beta'$-constraints can properly contain $F(\beta\text{-constraints})$.
\end{remark}

\begin{example}[Forgetful Functor $\mathbf{DCPO}\to\mathbf{Pos}$]
\label{ex:dcpo-to-pos}
The forgetful functor $U:\mathbf{DCPO}\to\mathbf{Pos}$ preserves finite limits: the terminal object is the one-point dcpo, products are componentwise dcpos, and the equalizer of Scott-continuous maps is the sub-dcpo on which they agree, whose underlying poset is the $\mathbf{Pos}$-equalizer. For a dcpo $P$ and global element $\rho\in P$, the $\mathbf{DCPO}$-constraint is the Scott-closed principal ideal $\downarrow\!\rho$, and the $\mathbf{Pos}$-constraint (in the enlarged class of all down-sets; see Remark~\ref{rem:pullback-stability}) is again $\downarrow\!\rho$. The comparison $\phi_{P,\rho}$ is the identity, so $U$ is strict. By Clearing Invariance (Theorem~\ref{thm:clearing-invariance}), passing from a Scott-continuous clearing model to its underlying monotone model leaves the global-section object unchanged.
\end{example}

The following theorem shows that clearing is invariant under change of coefficient category.  The conceptual core is short: since $H^0$ is a finite limit and $F$ preserves finite limits, $F$ transports the global-section object on the nose; the substantive content of the proof is the verification that the constraint subobjects $\beta(\liab)$ and the curried partial aggregator $\col_v=\widehat\col_v\circ(\iota_v\times\id)$ transport coherently along $F$.  Lossy comparisons -- scalarization, expectation, projection, generic risk measures -- typically fail this coherence; for lax morphisms, Remark~\ref{rem:lax-invariance} produces a canonical comparison morphism in place of an isomorphism.

\begin{theorem}[Clearing Invariance]
\label{thm:clearing-invariance}
Let $F:(\datacat,\mathcal S,\beta)\to(\datacat',\mathcal S',\beta')$ be a strict morphism of liability categories, and let $N=(G,\payob,\liab,\iota;\dis,\widehat\col)$ be a liability network in $\datacat$ with liability sheaf $\sheaf{L}$. The \style{pushforward network} $F_\ast N=(G, F\payob, F_\ast\liab, F_\ast\iota;\, F_\ast\dis, F_\ast\widehat\col)$ in $\datacat'$, defined by
\[
(F_\ast\payob)_v\ =\ F\payob_v,\qquad
(F_\ast\liab)_e\ =\ F\liab_e\circ\tau,\qquad
(F_\ast\iota)_v\ =\ F\iota_v\circ\tau,
\]
\[
(F_\ast\dis)_e\ =\ \phi_{\payob_{s(e)},\liab_e}\circ F\dis_e,\qquad
(F_\ast\widehat\col)_v\ =\ F\widehat\col_v\circ\kappa_v^{-1}\circ(\id_{F\payob_v}\times\psi_v^{-1}),
\]
is a liability network in $\datacat'$, where $\tau:1'\xrightarrow{\sim}F1$ is the canonical isomorphism induced by terminal-object preservation,
\[
\kappa_v:\ F\!\bigl(\payob_v\times\!\!\prod_{t(e)=v}\!\beta_{\payob_{s(e)}}(\liab_e)\bigr)\ \xrightarrow{\sim}\ F\payob_v\times F\!\!\prod_{t(e)=v}\!\beta_{\payob_{s(e)}}(\liab_e)
\]
is the canonical binary product comparison induced by finite-limit preservation of $F$, and
\[
\psi_v:\ F\!\!\prod_{t(e)=v}\!\beta_{\payob_{s(e)}}(\liab_e)\ \xrightarrow{\sim}\ \prod_{t(e)=v}\!\beta'_{F\payob_{s(e)}}\!\bigl((F_\ast\liab)_e\bigr)
\]
is the canonical isomorphism assembled from finite-limit preservation of $F$ and the strict comparisons $\phi$. Let $\sheaf{L}'$ denote the liability sheaf of $F_\ast N$, with total payment object $P'=\prod_{v\in V}F\payob_v$ and clearing operator $\Phi':P'\to P'$, and write $\chi_P:F\bigl(\prod_{v\in V}\payob_v\bigr)\xrightarrow{\sim}P'$ for the canonical product comparison induced by finite-limit preservation of $F$. Then there is a canonical isomorphism in $\datacat'$
\[
F\bigl(H^0(\hyper_G;\sheaf{L})\bigr)\ \cong\ H^0(\hyper_G;\sheaf{L}'),
\]
under which the clearing operators correspond in the sense that
\[
\Phi'\circ\chi_P\ =\ \chi_P\circ F\Phi.
\]
\end{theorem}

{\em Proof:}
We show $F\circ\sheaf{L}\cong\sheaf{L}'$ as functors $\Inc(\hyper_G)\to\datacat'$, then invoke limit preservation.

\style{Stalks.}  At an institution vertex $v\in V$,
$(F\sheaf{L})(v)=F\payob_v=\sheaf{L}'(v)$.
At a distribution hyperedge $h_v^{\dis}$,
$(F\sheaf{L})(h_v^{\dis})=F\payob_v=\sheaf{L}'(h_v^{\dis})$.
At a payment vertex $e^*$,
$(F\sheaf{L})(e^*)=F\bigl(\beta_{\payob_{s(e)}}(\liab_e)\bigr)$ and $\sheaf{L}'(e^*)=\beta'_{F\payob_{s(e)}}\!\bigl((F_\ast\liab)_e\bigr)$; by strictness, $\phi_{\payob_{s(e)},\liab_e}$ is an isomorphism identifying them.
At a collection hyperedge $h_v^{\col}$, finite-limit preservation gives $F\!\prod\beta\cong\prod F\beta$, and strictness gives $\prod F\beta\cong\prod\beta'$; their composite is the iso $\psi_v$ of the statement, identifying $(F\sheaf{L})(h_v^{\col})$ with $\sheaf{L}'(h_v^{\col})$.

\style{Restrictions.}  The identity restrictions $\sheaf{L}(h_v^{\dis}\to v)=\id_{\payob_v}$ are preserved by any functor. The distributor restrictions $\sheaf{L}(h_v^{\dis}\to e^*)=\dis_e$ push to $F\dis_e$, which equals $(F_\ast\dis)_e$ after post-composing with $\phi$; this is exactly the restriction $\sheaf{L}'(h_v^{\dis}\to e^*)$ under the stalk identifications. The projection restrictions $\sheaf{L}(h_v^{\col}\to e^*)=\pi_e$ are preserved because $F$ preserves products and $\psi_v$ is assembled from these projections. The aggregator restrictions $\sheaf{L}(h_v^{\col}\to v)=\col_v$ push to $F\col_v$.  The partial aggregator $(F_\ast\col)_v$ of the pushforward network, derived from $(F_\ast\widehat\col)_v$ and $(F_\ast\iota)_v$ as in Definition~\ref{def:liability-graph}, satisfies $F\col_v=(F_\ast\col)_v\circ\psi_v$ by direct computation from $\col_v=\widehat\col_v\circ(\iota_v\times\id)$, matching $\sheaf{L}'(h_v^{\col}\to v)$ after the identification.  Naturality of the resulting iso $F\sheaf{L}\cong\sheaf{L}'$ is immediate from the componentwise construction.

\style{Limit preservation.}  Since $F$ preserves finite limits and $\Inc(\hyper_G)$ is a finite category,
\[
F\bigl(H^0(\hyper_G;\sheaf{L})\bigr)
\ =\ F\bigl(\lim_{\Inc(\hyper_G)}\sheaf{L}\bigr)
\ \cong\ \lim_{\Inc(\hyper_G)}(F\sheaf{L})
\ \cong\ \lim_{\Inc(\hyper_G)}\sheaf{L}'
\ =\ H^0(\hyper_G;\sheaf{L}').
\]
For the operator equation, $F\Phi:FP\to FP$ and $\Phi':P'\to P'$ are intertwined by $\chi_P$ on the nose: at each component $v$, $\chi_P\circ F\Phi$ projects through the $v$-th component of $\chi_P$ to $F\col_v\circ F\!\!\prod\dis_e$, while $\Phi'\circ\chi_P$ does so to $(F_\ast\col)_v\circ\prod (F_\ast\dis)_e$, and the definitions $(F_\ast\col)_v=F\col_v\circ\psi_v^{-1}$ and $(F_\ast\dis)_e=\phi_{\payob_{s(e)},\liab_e}\circ F\dis_e$ together with strictness identify the two. 
The statement on $\Fix(\Phi_\ast)$ then follows by applying $\Hom_{\datacat'}(1',-)$ to both sides and using that $F$ carries the equalizer of $(\id_P,\Phi)$ in $\datacat$ to the equalizer of $(\id_{P'},\Phi')$ in $\datacat'$ (equalizers being finite limits).
\qed

\begin{corollary}[Clearing Configurations]
\label{cor:clearing-configs}
Under the hypotheses of Theorem~\ref{thm:clearing-invariance}, the canonical map
\[
\Hom_\datacat(1,-)\ \longrightarrow\ \Hom_{\datacat'}(1',F(-))
\]
restricts to a map $\Gamma(\hyper_G;\sheaf{L})\to\Gamma(\hyper_G;\sheaf{L}')$. If $F$ is \style{global-element bijective}, meaning that the canonical map $\Hom_\datacat(1,X)\to\Hom_{\datacat'}(1',FX)$ is a bijection for $X=H^0(\hyper_G;\sheaf{L})$, then this restriction is a bijection
\[
\Gamma(\hyper_G;\sheaf{L})\ \xrightarrow{\ \sim\ }\ \Gamma(\hyper_G;\sheaf{L}'),
\]
and accordingly $\Fix(\Phi_\ast)\cong\Fix(\Phi'_\ast)$ as sets.
\end{corollary}

{\em Proof:} Apply $\Hom_\datacat(1,-)$ and $\Hom_{\datacat'}(1',-)$ to the isomorphism of Theorem~\ref{thm:clearing-invariance}. Naturality of the comparison $\Hom_\datacat(1,X)\to\Hom_{\datacat'}(1',FX)$ in $X$ gives the asserted restriction; bijectivity at $X=H^0(\hyper_G;\sheaf{L})$ promotes the restriction to a bijection. 
The identification with $\Fix(\Phi_\ast)$, $\Fix(\Phi'_\ast)$ is Proposition~\ref{prop:limit-equalizer}. \qed

\begin{remark}[Lax Case]
\label{rem:lax-invariance}
When $F$ is a lax (non-strict) morphism, the monomorphisms $\phi_{\payob,\rho}$ need not be isomorphisms, and the product isomorphism $\psi_v$ is only a monomorphism. The pushforward network $F_\ast N$ is then not canonically defined in $\datacat'$: the aggregator $F\col_v$ lives on $F\prod\beta$, a proper subobject of $\prod\beta'$, and any extension to the full domain $\prod\beta'$ must be chosen as additional data. For any such choice making the pushforward a liability network, Lemma \ref{lem:fixed-base-functoriality} produces a canonical morphism 
\[
F\bigl(H^0(\hyper_G;\sheaf{L})\bigr)\ \longrightarrow\ H^0(\hyper_G;\sheaf{L}'),
\]
but this is an isomorphism only in the strict case.
\end{remark}

\begin{remark}[Applications of Clearing Invariance]
\label{rem:invariance-applications}
Theorem~\ref{thm:clearing-invariance} applies uniformly to every strict comparison of clearing semantics that arises as a finite-limit-preserving, constraint-isomorphism functor. Representative cases include:
\begin{itemize}
\item \style{Recovery of Eisenberg--Noe.} The bounded Eisenberg--Noe model may be treated directly as a $\mathbf{Pos}$-valued liability sheaf with payment objects $[0,\bar\ell_v]$ and constraint subobjects $[0,\ell_e]\hookrightarrow[0,\bar\ell_v]$.  Its inclusion into the unbounded presentation with payment objects $[0,\infty]$ induces an isomorphism on $H^0$, since the clearing operator carries $\prod_v[0,\infty]$ into $\prod_v[0,\bar\ell_v]$ after a single iteration (Remark~\ref{rem:EN-bounded-presentation}).  Thus the classical scalar Eisenberg--Noe model is the global-element shadow of either presentation (Corollary~\ref{cor:clearing-configs}).
\item \style{Change of basis.} Currency redenomination by a positive scalar acts as an order-isomorphism on each interval payment object and assembles into a strict equivalence of liability categories; clearing sections transport canonically and commute with redenomination.
\end{itemize}
Each application is a single invocation of the theorem; no separate proof is required.  The scope of Clearing Invariance is exactly change of coefficient category over a \emph{fixed} base hypergraph: it is not a network-morphism theorem.  Other natural changes of semantics -- subnetwork restriction, vertex merging, multilateral netting, and central-counterparty interposition -- change the base hypergraph itself, and require comparison maps (functors, spans, or profunctors) between clearing-section lattices of related hypergraphs rather than an isomorphism over a single base; these are taken up in forthcoming work.
\end{remark}

\section{Current Clearing Models as Liability Sheaves}
\label{sec:current}

We now show that Eisenberg--Noe and lattice liability networks arise as special cases, obtained by specializing the choice of liability category and the distributor and aggregator morphisms.  The fixed-point methods of Section~\ref{sec:fixed-point} apply uniformly to both.

\subsection{Classical Eisenberg-Noe}
\label{ssec:E-N}

We demonstrate how the classical Eisenberg-Noe model emerges as a special case of our categorical framework. This example serves both to validate our abstraction and to illustrate the concrete mechanics of the sheaf construction.

Consider a financial network with institutions $V = \{1, 2, \ldots, n\}$ and directed edges $E$ representing debt obligations. Each institution $i$ has external assets $c_i \geq 0$ and owes nominal amounts $\ell_{ij} > 0$ to other institutions; throughout this subsection, $E$ contains exactly those pairs $(i,j)$ with strictly positive nominal liability, so that every edge $e\in E$ has $\ell_e>0$.  We model this system as follows.

\textbf{Data Category.} We work in the category $\mathbf{Pos}$ of partially ordered sets with monotone maps. Each payment object $\payob_v = [0, \infty]$ is a complete lattice object in the sense of Definition~\ref{def:complete-lattice-object} below: $\Hom(1, \payob_v) = \payob_v$ is a complete lattice. The category $\mathbf{Pos}$ satisfies the requirements of Definition~\ref{def:liability-category}: the terminal object is the one-element poset $1 = \{\star\}$, products are Cartesian products with componentwise order, the distinguished constraint class consists of all down-sets (Remark~\ref{rem:pullback-stability}), and equalizers are sub-posets with induced order. The bound selector $\beta$ produces the principal down-sets $\payob_{s(e)}^{\lambda_e} = \{x \in \payob_{s(e)} : x \leqslant \ell_e\}$ used as constraint subobjects at payment vertices. The order-theoretic fixed-point arguments (via Tarski) take place within the complete lattice structure of $\prod_{v \in V} \payob_v$.

\textbf{Liability Network.} We define $(G, \payob, \liab, \iota)$ where $G = (V, E, s, t)$ is the directed graph of debt relationships, each payment object $\payob_v = [0,\infty]$ as fixed in the Data Category above, each liability morphism $\liab_e: 1 \to [0,\infty]$ maps $\star \mapsto \ell_e$ (the nominal liability along edge $e$), and each exogenous resource morphism $\iota_v: 1 \to [0,\infty]$ maps $\star \mapsto c_v$ (the external assets of institution $v$).

\textbf{Liability Sheaf.} The sheaf $\sheaf{L}: \Inc(\hyper_G) \to \mathbf{Pos}$ assigns stalks:
\begin{itemize}
\item $\stalk{\sheaf{L}}{v} = [0,\infty]$ for each institution $v$
\item $\stalk{\sheaf{L}}{e^*} = [0, \ell_e]$ for each payment vertex $e^*$
\item $\stalk{\sheaf{L}}{h_v^{\dis}} = [0,\infty]$ for each distribution hyperedge
\item $\stalk{\sheaf{L}}{h_v^{\col}} = \prod_{t(e) = v} [0, \ell_e]$ for each collection hyperedge
\end{itemize}

For the morphisms, let $\bar{\ell}_v = \sum_{s(e) = v} \ell_e$ denote the total liabilities of institution $v$. When $\bar{\ell}_v = 0$, institution $v$ has no outgoing liabilities and we set $\dis_e \equiv 0$ for all $e$ with $s(e) = v$. When $\bar{\ell}_v > 0$, the distributor $\dis_e: [0,\infty] \to \payob_{s(e)}^{\lambda_e} = [0, \ell_e]$ implements proportional payment with respect to nominal liability bounds:
\[
\dis_e(x) = \min\left\{ x \cdot \frac{\ell_e}{\bar{\ell}_v}, \ell_e \right\}
\]

The aggregator $\col_v: \prod_{t(e) = v} [0, \ell_e] \to [0,\infty]$ is defined as:
\[
\col_v\left((p_e)_{t(e) = v}\right) = \widehat{\col}_v\left(c_v, (p_e)_{t(e) = v}\right) = \min\left\{ \bar{\ell}_v, c_v + \sum_{t(e) = v} p_e \right\}
\]
where $c_v = \iota_v(\star)$ are the external assets and $\widehat{\col}_v(r, (p_e)) = \min\{\bar{\ell}_v, r + \sum p_e\}$ is the underlying aggregation.
When $t^{-1}(v) = \emptyset$, the empty sum equals zero and $\col_v(\star) = \min\{\bar{\ell}_v, c_v\}$.

\textbf{Clearing Characterization.} By Theorem \ref{thm:clearing-sections}, a global section corresponds to institution payments $x_v$ and edge payments $p_e$ satisfying:
\[
p_e = \min\left\{ x_{s(e)} \cdot \frac{\ell_e}{\bar{\ell}_{s(e)}}, \ell_e \right\} 
\qquad
x_v = \min\left\{ \bar{\ell}_v, c_v + \sum_{t(e) = v} p_e \right\} 
\]
with substitution yielding the classical Eisenberg-Noe clearing equation:
\[
x_v = \min\left\{ \bar{\ell}_v, c_v + \sum_{t(e) = v} \min\left\{ x_{s(e)} \cdot \frac{\ell_e}{\bar{\ell}_{s(e)}}, \ell_e \right\} \right\}
\]
The outer minimum enforces limited liability while the inner minimum ensures individual payments respect nominal bounds.

\textbf{Existence via Tarski.} All morphisms are monotone on complete lattices, so existence of clearing sections is guaranteed by the Tarski Fixed Point Theorem (see Theorem \ref{thm:lattice-existence} below). Therefore $\Gamma(\hyper_G; \sheaf{L})$ is non-empty and forms a complete lattice, with the greatest element corresponding to the maximal clearing vector.

\textbf{Computational Aspects.} The operator
$\Phi:\prod_v[0,\infty]\to\prod_v[0,\infty]$ from
Theorem~\ref{thm:lattice-existence} takes the form
\[
[\Phi(\mathbf{x})]_v
= \min\left\{\bar{\ell}_v,\,
  c_v+\sum_{t(e)=v}
  \min\left\{x_{s(e)}\cdot
  \frac{\ell_e}{\bar{\ell}_{s(e)}},\,\ell_e\right\}
  \right\}.
\]
The outer $\min\{\bar\ell_v,\,\cdot\,\}$ shows that $\Phi$
carries $\prod_v[0,\infty]$ into the bounded complete
sublattice $\prod_v[0,\bar\ell_v]$, so clearing dynamics
are confined to this sublattice after a single
application of $\Phi$. Restricted to it, $\Phi$ is a
composition of piecewise-linear monotone functions on a
compact complete lattice and therefore preserves both
directed suprema and filtered infima.  Kleene
iteration (Remark~\ref{rem:iteration}) therefore converges
in both directions: starting from
$\mathbf{x}^{(0)}=\mathbf{0}$ yields an increasing
sequence converging to the least clearing vector, while
starting from
$\mathbf{x}^{(0)}=(\bar{\ell}_v)_{v\in V}$ yields a
decreasing sequence converging to the greatest clearing
vector.

\begin{remark}[Bounded versus Unbounded Presentations]
\label{rem:EN-bounded-presentation}
The body's choice $\payob_v=[0,\infty]$ is one of two natural presentations of the Eisenberg--Noe model.  The \style{bounded presentation} takes $\payob_v=[0,\bar\ell_v]$ in the smallest full subcategory $\datacat_{\mathrm{EN}}\subseteq\mathbf{Pos}$ closed under finite limits and containing $\{[0,\bar\ell_v]:v\in V\}$ and $\{[0,\ell_e]:e\in E\}$; in this presentation the aggregator $\widehat\col_v(c_v,(p_e))=\min\{\bar\ell_v,\,c_v+\sum p_e\}$ lands directly in the payment object and the outer truncation is absorbed into the codomain.  Because the exogenous resource morphism $\iota_v:1\to\payob_v$ must take values in $[0,\bar\ell_v]$, the bounded presentation uses $\iota_v(\star)=\min\{c_v,\bar\ell_v\}$ in place of $c_v$; the outer truncation in $\widehat\col_v$ makes this substitution invisible at the level of $\Phi$ and of global sections.  The body's \style{unbounded presentation} is more uniform -- the payment object $[0,\infty]$ does not depend on the network data -- but requires the outer $\min$ in $\Phi$ to enforce the bound. Because $\Phi$ stabilizes $\prod_v[0,\bar\ell_v]\subseteq\prod_v[0,\infty]$ after a single iteration (paragraph above), the two presentations have canonically isomorphic global-section objects: every clearing section of the unbounded model already lies in the bounded sublattice, and the natural inclusion of liability sheaves (Lemma~\ref{lem:fixed-base-functoriality}) restricts to an isomorphism on $H^0$. The Recovery of Eisenberg--Noe application of Remark~\ref{rem:invariance-applications} invokes the bounded presentation.
\end{remark}

\begin{remark}[The Sheaf Framework and Classical Theory]
The bounded stalks $\payob_{s(e)}^{\lambda_e} = [0, \ell_e]$ at payment vertices automatically ensure edge payments never exceed nominal liabilities, eliminating explicit constraints in the clearing equations.  Though each institution's clearing condition is purely local -- balancing inflows and outflows -- the sheaf structure enforces global consistency through compatibility at hyperedges.

The pointwise ordering on clearing vectors has a natural interpretation: $\mathbf{x} \leqslant \mathbf{y}$ means $x_v \leqslant y_v$ for all $v$, so higher vectors represent less default.  The least clearing section corresponds to maximal default propagation, the greatest to optimal recovery.  This lattice structure follows from Tarski's theorem applied to the monotone clearing operator, and is independent of the specific financial model.
\end{remark}

\subsection{Lattice Liability Networks}
\label{ssec:LLNs}

The lattice liability networks (LLNs) of \cite{GhristGouldLopezRiess2025} fit within our categorical framework. This example demonstrates both the power of the sheaf formulation and highlights a subtle issue regarding resource conservation.

Recall that an LLN consists of payment lattices at vertices, nominal liabilities along edges, and operators for aggregating incoming resources and distributing outgoing payments. The key insight is that these structures naturally correspond to liability sheaves when working in an appropriate category of lattices.

We regard each payment lattice $L_v$ as a complete lattice object of $\mathbf{Pos}$ in the sense of Definition~\ref{def:complete-lattice-object}: $\Hom_{\mathbf{Pos}}(1,L_v)=L_v$ inherits the lattice order and is complete. When the LLN structure maps are Scott-continuous, the same construction refines to $\mathbf{DCPO}$.

\textbf{Liability Data.} For each edge $e \in E$ with source $v = s(e)$, the liability morphism $\liab_e: 1 \to L_{s(e)}$ picks out the nominal liability element $\ell_e \in L_{s(e)}$, with liability datum $\lambda_e = \liab_e(\star) = \ell_e$. The corresponding constraint subobject is the principal ideal $L_{s(e)}^{\lambda_e} = \{x \in L_{s(e)} : x \leqslant \ell_e\}$.

\textbf{LLN as Liability Network.} Given an LLN on quiver $Q = (V, E, s, t)$, we construct a liability network $(G, \payob, \liab, \iota)$ where $G = Q$, each payment object $\payob_v$ is the complete lattice $L_v$ from the LLN, regarded as a complete lattice object of $\mathbf{Pos}$, each liability morphism $\liab_e: 1 \to L_{s(e)}$ picks out the nominal liability element $\ell_e \in L_{s(e)}$, and each exogenous resource morphism $\iota_v: 1 \to L_v$ selects the exogenous resource element in $L_v$.

\textbf{Bounded Subobjects.} For each edge $e$ with source $v = s(e)$ and nominal liability $\ell_e \in L_v$, the constraint subobject is the principal ideal: $\downarrow\!\ell_e = \{x \in L_{s(e)} : x \leqslant \ell_e\}$.
This is a complete lattice with the induced order, and the inclusion $L_{s(e)}^{\lambda_e} \hookrightarrow L_{s(e)}$ is monotone.

\textbf{Distributor Decomposition.} The original LLN framework uses a single distributor $\distrib_v: L_v \to \prod_{e \in s^{-1}(v)} L_v$. In the sheaf framework, we need individual distributor morphisms $\dis_e: L_v \to L_{s(e)}^{\lambda_e}$ for each outgoing edge $e$ where $s(e) = v$.

Let $\pi_e: \prod_{e' \in s^{-1}(v)} L_v \to L_v$ be the projection onto the component for edge $e$. Since the original framework requires $[\distrib_v(\top_v)]_e \leqslant \ell_e$ and $\distrib_v$ is monotone, the image of $\pi_e \circ \distrib_v$ lies in the principal ideal $\downarrow\!\ell_e = L_{s(e)}^{\lambda_e}$. Thus, the sheaf distributor is simply:
\[
\dis_e = \pi_e \circ \distrib_v : L_v \to L_{s(e)}^{\lambda_e}
\]
where we restrict the codomain to $L_{s(e)}^{\lambda_e}$. This defines a morphism with the required codomain in the sheaf framework.

\textbf{Aggregator Correspondence.} The LLN pay-in aggregator $\inagg_v$ has domain $\prod_{e \in t^{-1}(v)} L_{s(e)}$, while the sheaf aggregator $\col_v$ has domain $\prod_{e \in t^{-1}(v)} L_{s(e)}^{\lambda_e}$, the product of bounded stalks. With $j$ the canonical inclusion of the latter into the former, we set
\[
\col_v = \inagg_v \circ j : \prod_{e \in t^{-1}(v)} L_{s(e)}^{\lambda_e} \to L_v ,
\]
inheriting from $\inagg_v$ the implicit incorporation of exogenous resources (Remark~\ref{rem:exogenous}).

\textbf{Conservation and Factorization.} The original LLN framework imposes the conservation requirement $\outagg_v \circ \distrib_v = \id_{L_v}$, where $\outagg_v: L_v^{|s^{-1}(v)|} \to L_v$ is the pay-out aggregator. This ensures that resources distributed by $\distrib_v$ can be perfectly reconstructed by $\outagg_v$.

In the sheaf framework, this conservation property is not automatically preserved. The individual distributor morphisms $\dis_e$ extract components of the full distribution, but there is no explicit mechanism to ensure these components can be reassembled to recover the original payment state. This reveals a fundamental difference between the two frameworks:

\begin{itemize}
\item The LLN framework enforces strict resource conservation at each vertex through the factorization requirement
\item The sheaf framework enforces compatibility conditions between vertices but allows more flexibility in local resource handling
\end{itemize}

\begin{remark}[Conservation as Additional Structure]
One could impose conservation as an additional requirement on liability sheaves by demanding that for each vertex $v$, there exists a pay-out aggregator $\outagg_v: \prod_{e \in s^{-1}(v)} L_{s(e)}^{\lambda_e} \to L_v$ that serves as a left-inverse to the collective distribution map. Specifically, if we define the collective distributor $\mathcal{D}_v: L_v \to \prod_{e \in s^{-1}(v)} L_{s(e)}^{\lambda_e}$ by $\mathcal{D}_v(x) = (\dis_e(x))_{e \in s^{-1}(v)}$, then conservation requires:
\[
\outagg_v \circ \mathcal{D}_v = \id_{L_v}
\]
This condition states that no ``value'' is lost in the distribution step. However, conservation in this algebraic sense is distinct from conservation of money. Even in the classical Eisenberg--Noe model, defaulting institutions pay strictly less than they owe; the difference is a loss absorbed by creditors. The relevant economic requirement is not $\outagg_v \circ \mathcal{D}_v = \id_{L_v}$ but rather that the clearing operator be well-defined and monotone, which the sheaf framework ensures through the distributor and aggregator morphisms directly. For non-financial applications (permission networks, information flow), resources need not be conserved at all. When conservation is desired, it can be imposed as additional structure on the liability sheaf without modifying the framework's foundations.
\end{remark}

\textbf{Clearing Correspondence.} Despite dropping the conservation axiom, the sheaf framework recovers the LLN clearing \emph{equations}: clearing sections in the two frameworks are in bijection. Conservation can be imposed as additional structure on the liability sheaf (see the remark above) but is not required for the clearing conditions themselves.
A global section of the liability sheaf consists of payments $x_v \in L_v$ satisfying:
\begin{align*}
p_e &= \dis_e(x_{s(e)}) = [\distrib_{s(e)}(x_{s(e)})]_e \\
x_v &= \col_v\left((p_e)_{t(e) = v}\right) = (\inagg_v \circ j)\left((p_e)_{t(e) = v}\right)
\end{align*}

Since $[\distrib_v(x_v)]_e \leqslant \ell_e$ holds in the LLN framework, $\dis_e = \pi_e \circ \distrib_v$ lands in $L_{s(e)}^{\lambda_e}$ as required, and the equations reduce to $p_e = [\distrib_{s(e)}(x_{s(e)})]_e$ and $x_v = \inagg_v((p_e)_{t(e) = v})$ -- the LLN clearing conditions.

\begin{remark}[Role of the Sheaf Formulation]
The sheaf framework does not change the fixed-point content of the clearing problem: the operator $\Phi$ on $\prod \payob_v$ could be defined directly. What the sheaf provides is a \emph{structured decomposition} of this operator into local data (stalks, distributors, aggregators) governed by the hypergraph topology. This decomposition has several concrete consequences:
\begin{itemize}
\item \textbf{Separation of concerns:} Distribution and collection are handled by distinct hyperedges, making the flow structure explicit and constraining the class of operators under study.
\item \textbf{Automatic bounds:} The bounded stalks $L_{s(e)}^{\lambda_e}$ at payment vertices ensure liability constraints by construction, rather than by explicit verification for each model.
\item \textbf{Categorical generality:} The same framework accommodates lattice-valued and other structured payments through the choice of liability category, with order- or metric-based fixed-point guarantees imposed separately on the set of global elements; see Section~\ref{sec:fixed-point}.
\end{itemize}
The main theorems of Section~\ref{sec:fixed-point} apply to the clearing operator $\Phi$; the sheaf structure organizes the \emph{hypotheses} under which those theorems hold.
\end{remark}

\section{Fixed Point Theorems}
\label{sec:fixed-point}

The liability-sheaf construction characterizes clearing configurations as global sections but does not by itself guarantee their existence: the constraints $p_e=\dis_e(x_{s(e)})$ and $x_v=\col_v((p_e)_{t(e)=v})$ may admit no consistent assignment, particularly when cycles propagate compatibility requirements around loops.  Four standard fixed-point principles, applied to the operator $\Phi_\ast$ on global elements, supply existence and -- in the acyclic and metric cases -- uniqueness.  In each case the work is done by the chosen hypothesis on payment objects or on the metric structure of $\mathcal G_P=\Hom_\datacat(1,P)$; the sheaf supplies the factorization $\Phi=A\circ D$ and the canonical bijection $\Gamma(\hyper_G;\sheaf{L})\cong\Fix(\Phi_\ast)$ of Proposition~\ref{prop:limit-equalizer}.

\begin{definition}[Complete Lattice Object]
\label{def:complete-lattice-object}
An object $L$ in a liability category $\datacat$ is a \style{complete lattice object} if $\Hom(1,L)$ forms a complete lattice under the partial order from Definition~\ref{def:liability-category}.  In $\mathbf{Pos}$, complete lattice objects are precisely the complete lattices; in $\mathbf{DCPO}$, they are the dcpos whose underlying poset is a complete lattice; in $\mathbf{Set}_{\mathrm{disc}}$, only singletons qualify.
\end{definition}

\begin{theorem}[Existence via Tarski]
\label{thm:lattice-existence}
Let $\sheaf{L}$ be a liability sheaf on $\hyper_G$ valued in a liability category $\datacat$ in which each payment object $\payob_v$ is a complete lattice object.  Then $\Gamma(\hyper_G;\sheaf{L})$ is non-empty and forms a complete lattice under the pointwise order $\sigma\leqslant\tau\Leftrightarrow\sigma_v\leqslant\tau_v$ for all $v\in V$.
\end{theorem}

{\em Proof:} By product-order compatibility (axiom~(5) of Definition~\ref{def:liability-category}), $\mathcal G_P=\Hom(1,P)\cong\prod_v\Hom(1,\payob_v)$ is a complete lattice under componentwise order.  Order-preservation of each $\dis_e$ and $\col_v$ on global elements (axiom~(2)) makes $\Phi_\ast(\mathbf{x})_v=\col_v((\dis_e(x_{s(e)}))_{t(e)=v})$ monotone, and Tarski's theorem yields a non-empty complete lattice $\Fix(\Phi_\ast)$.  Proposition~\ref{prop:limit-equalizer} identifies $\Gamma(\hyper_G;\sheaf{L})$ with $\Fix(\Phi_\ast)$ as a set; the bijection sends $\mathbf{x}\mapsto(\sigma_v=x_v,\,\sigma_{e^*}=\dis_e(x_{s(e)}))$ and is an order isomorphism for the pointwise order on sections, since each $\dis_e$ is monotone.  The lattice operations on $\Gamma(\hyper_G;\sheaf{L})$ are the Tarski lattice operations transported from $\Fix(\Phi_\ast)$; they need not be computed by coordinatewise joins and meets in the ambient product. \qed

\begin{remark}[Completeness Cannot Be Dropped]
\label{rem:completeness-needed}
A simple multiplicative cycle of currency conversions with round-trip product strictly greater than one yields a monotone clearing operator on the non-complete lattice $[1,\infty)^n$ with no fixed point there, while its completion $[1,\infty]^n$ admits only the vacuous solution $(\infty,\ldots,\infty)$.  Thus amplification around a cycle can obstruct the existence of an economically meaningful clearing section, and completeness is the order-theoretic hypothesis that restores formal existence.
\end{remark}

\begin{remark}[Convergent Iteration]
\label{rem:iteration}
Tarski's theorem is non-constructive: extremal fixed points are presented as infima and suprema over pre- and post-fixed sets.  When each $\payob_v$ has finite height, the iterates $\Phi_\ast^n(\bot)$ stabilize in at most $\sum_v h(\payob_v)$ steps.  More generally, if the product lattice satisfies the ascending chain condition, the ascending iteration stabilizes, though without a uniform finite-height bound.  Without such a height hypothesis, convergence of $\Phi_\ast^n(\bot)$ to the least clearing section requires Scott continuity, addressed in Theorem~\ref{thm:kleene} below.
\end{remark}

\subsection{Convergent Iteration via Scott Continuity}
\label{ssec:kleene}

In the Scott-continuous regime -- naturally captured by liability sheaves valued in $\mathbf{DCPO}$ -- Kleene iteration converges to the extremal clearing sections.

\begin{theorem}[Kleene Iteration]
\label{thm:kleene}
Let $\sheaf{L}$ be a liability sheaf on $\hyper_G$ valued in $\mathbf{DCPO}$, with each payment object $\payob_v$ a complete-lattice object, and let $\bot,\top\in\mathcal G_P$ be the bottom and top of the product lattice.  Let $\mathbf{x}_{\min},\mathbf{x}_{\max}\in\mathcal G_P$ denote the institution-coordinate projections of the least and greatest clearing sections (Theorem~\ref{thm:lattice-existence}).  Then $\Phi_\ast$ is Scott-continuous and
\[
\mathbf{x}_{\min}\ =\ \sup_{n\ge 0}\Phi_\ast^n(\bot).
\]
If, in addition, $\Phi_\ast$ preserves filtered infima -- as it does whenever each $\payob_v$ has finite height or $\mathcal G_P$ is a finite product of compact real intervals on which $\dis_e,\col_v$ are continuous in the Euclidean topology -- then
\[
\mathbf{x}_{\max}\ =\ \inf_{n\ge 0}\Phi_\ast^n(\top).
\]
\end{theorem}

{\em Proof:} Each $\dis_e$ and $\col_v$ is Scott-continuous by hypothesis, so $\Phi=A\circ D$ and its action $\Phi_\ast$ on global elements are Scott-continuous.  The Kleene fixed-point theorem identifies $\sup_n\Phi_\ast^n(\bot)$ with the least pre-fixed point above $\bot$, which by Tarski is $\mathbf{x}_{\min}$.  Under filtered-infimum preservation, $y=\inf_n\Phi_\ast^n(\top)$ satisfies $\Phi_\ast(y)=\inf_n\Phi_\ast^{n+1}(\top)=y$, and monotonicity gives $z\leqslant\Phi_\ast^n(\top)$ for any other fixed point $z$, hence $z\leqslant y$; thus $y=\mathbf{x}_{\max}$.  Verification of the two sufficient conditions for filtered-infimum preservation is standard: in the finite-height case every descending chain stabilizes; in the compact-interval case monotone descending nets converge coordinatewise to their infima, and continuity transports this through $\Phi_\ast$. \qed

\begin{example}[Eisenberg--Noe redux]
\label{ex:EN-kleene}
The classical Eisenberg--Noe operator on $\prod_v[0,\bar\ell_v]$ is monotone and piecewise linear, hence preserves directed suprema and filtered infima.  Theorem~\ref{thm:kleene} therefore yields convergence in both directions: from $\bot=\mathbf{0}$ to the least clearing vector and from $\top=(\bar\ell_v)$ to the greatest -- the same extremal selected by the fictitious-default algorithm of \cite{EisenbergNoe2001}.  The two procedures are not identical: Kleene iteration is a general monotone fixed-point iteration, while fictitious-default uses the piecewise-linear structure of this particular operator to identify the default set in $|V|$ combinatorial rounds and accelerate the descending computation.
\end{example}

\subsection{Acyclic Networks}
\label{ssec:dag}

When the underlying graph is acyclic, the structural difficulty addressed by Tarski, Kleene, and Banach -- cyclic propagation of compatibility constraints -- disappears, and a unique clearing section is computed by forward propagation in any liability category.

\begin{theorem}[Clearing on Acyclic Networks]
\label{thm:dag-clearing}
Let $\sheaf{L}$ be a liability sheaf on $\hyper_G$ valued in a liability category $\datacat$, with $G$ acyclic, and let $r$ denote the maximum length, in edges, of a directed path in $G$.  Then $\Phi_\ast^{r+1}$ is constant on $\mathcal G_P$, with image a single global element $\mathbf{x}^\ast\in\Fix(\Phi_\ast)$; equivalently, $\hyper_G$ admits a unique clearing section, computed by $r+1$ applications of $\Phi_\ast$ to any starting point.
\end{theorem}

{\em Proof:} Define the \style{level} of $v\in V$ as the maximum length of a directed path ending at $v$, so that $r=\max_v\mathrm{level}(v)$.  The expression $[\Phi_\ast(\mathbf{x})]_v=\col_v((\dis_e(x_{s(e)}))_{t(e)=v})$ depends only on coordinates $x_u$ with $\mathrm{level}(u)<\mathrm{level}(v)$.

By induction on $\ell=\mathrm{level}(v)$, the coordinate $[\Phi_\ast^k(\mathbf{x})]_v$ is independent of $\mathbf{x}$ whenever $k\ge\ell+1$.  At $\ell=0$, $t^{-1}(v)\cap E=\emptyset$ and $[\Phi_\ast(\mathbf{x})]_v=\col_v(\star)$ is constant.  For $\ell\ge 1$ and $k\ge\ell+1$, the value depends on $\{[\Phi_\ast^{k-1}(\mathbf{x})]_u:u\to v\}$, each of which has level $\leqslant\ell-1\leqslant k-2$ and is constant by induction.

Taking $k=r+1$, every coordinate of $\Phi_\ast^{r+1}(\mathbf{x})$ is independent of $\mathbf{x}$; call the resulting element $\mathbf{x}^\ast$.  One further iteration leaves $\mathbf{x}^\ast$ unchanged, so $\mathbf{x}^\ast\in\Fix(\Phi_\ast)$; and any $\mathbf{y}\in\Fix(\Phi_\ast)$ equals $\Phi_\ast^{r+1}(\mathbf{y})=\mathbf{x}^\ast$.  Proposition~\ref{prop:limit-equalizer} then gives $\Gamma(\hyper_G;\sheaf{L})=\{\mathbf{x}^\ast\}$. \qed

\begin{remark}[No Lattice or Metric Required]
\label{rem:dag-no-structure}
Theorem~\ref{thm:dag-clearing} uses neither order nor metric on payment objects, only graph topology and the empty-product aggregator convention.  It applies to liability sheaves valued in any liability category -- including $\mathbf{Set}_{\mathrm{disc}}$, where the clearing section is the unique forward propagation of payments through the acyclic network, a deterministic computation rather than a fixed-point search.  Existence, uniqueness, and contagion behavior are concentrated in the cycle structure of $G$, not in the algebraic structure of payment objects.
\end{remark}

\subsection{Uniqueness via Contraction on Global Elements}
\label{ssec:banach}

A complementary route to uniqueness imposes a complete metric on $\mathcal G_P$ rather than order structure on $\datacat$.  This cleanly separates the categorical and analytic layers: the sheaf lives in $\datacat$ (possibly order-trivial), and contraction is required only on global elements.

\begin{theorem}[Uniqueness via Banach]
\label{thm:banach-global}
Let $\sheaf{L}$ be a liability sheaf on $\hyper_G$ valued in a liability category $\datacat$, and suppose $\mathcal G_P$ carries a complete metric $d$ under which $\Phi_\ast$ is $k$-Lipschitz with $k<1$.  Then $\Fix(\Phi_\ast)=\{\mathbf{x}^\ast\}$, with $\mathbf{x}^\ast=\lim_n\Phi_\ast^n(\mathbf{x}_0)$ for every $\mathbf{x}_0\in\mathcal G_P$, and $\Gamma(\hyper_G;\sheaf{L})$ consists of a single clearing section.
\end{theorem}

{\em Proof:} Apply the Banach fixed point theorem to $(\mathcal G_P,d,\Phi_\ast)$ and pull back through Proposition~\ref{prop:limit-equalizer}. \qed

\begin{remark}[The $\ell^1$ Sufficient Condition]
\label{rem:ell-one-contraction}
If each $\payob_v$ carries a complete metric $d_v$ on its global elements with $\mathcal G_P$ equipped with $d(\mathbf{x},\mathbf{y})=\sum_v d_v(x_v,y_v)$, and each $\dis_e$ is $k_e$-Lipschitz and each $\col_v$ has componentwise Lipschitz constants $L_{v,e}$, the standard $\ell^1$ estimate
\[
d(\Phi_\ast(\mathbf{x}),\Phi_\ast(\mathbf{y}))\ \leqslant\ \bigl(\textstyle\max_w\sum_{s(e)=w}L_{t(e),e}k_e\bigr)\,d(\mathbf{x},\mathbf{y})
\]
gives contraction whenever the bracketed quantity is below $1$.  Stochastic clearing via Wasserstein contraction on probability-valued global elements fits this template and awaits future development.  In default-cost models such as Rogers--Veraart \cite{RogersVeraart2013}, recovery factors below one suggest natural contraction, but the standard clearing map is piecewise-defined across the solvent/default boundary and generally fails to be globally Lipschitz; attenuated variants that regularize this transition recover contraction.  The classical Eisenberg--Noe model is the non-expansive boundary case, where Tarski rather than Banach is the appropriate tool.
\end{remark}

\section{Outlook}
\label{sec:outlook}

This paper has done one mathematical job: it has identified liability clearing as a finite-limit construction in a liability category.  Clearing configurations are global sections of a liability sheaf on a directed hypergraph; the global-section object is the equalizer of the identity and a factored clearing operator $\Phi=A\circ D$ (Proposition~\ref{prop:limit-equalizer}); and this construction is functorial under change of coefficient category (Theorem~\ref{thm:clearing-invariance}).  Tarski's theorem is the order-theoretic corollary; Banach's theorem on global elements is the analytic one.  The Eisenberg--Noe model and lattice liability networks fall out as special cases.

The structural payoff is the same in every instance: clearing -- a seemingly opaque fixed-point problem -- decomposes canonically into local distribution and aggregation steps governed by the hypergraph topology, and this decomposition transports coherently across semantic frameworks.  That the same categorical skeleton subsumes deterministic and order-theoretic models suggests ``clearing'' is not fundamentally about money but about global consistency of local operations, a principle applicable to resource allocation, network equilibria, and coordination problems more broadly.

The framework invites extension along several directions, sketched below. 

\style{Changes of base hypergraph.}  Where Theorem~\ref{thm:clearing-invariance} addresses change of coefficient category over a fixed base, many of the most consequential operations on liability systems -- central-counterparty interposition, multilateral netting, subnetwork restriction, the gluing of overlapping clearing systems, and vertex merging more generally -- change the hypergraph itself.  Such operations are not expected to be governed by a single isomorphism theorem but by comparison maps between the clearing-section lattices of related hypergraphs, with the precise categorical form (functor, span, adjunction, or profunctor) depending on the operation.  Aggregation of obligations along an edge requires more structure on the liability category than the finite-limit axioms developed here, and restriction and gluing rely on a notion of open or boundaried hypergraph that the present work does not formalize.  These constructions are taken up in companion projects.

\style{Probabilistic and metric clearing.}  Equipping the set of global elements of each payment object with a metric -- the Wasserstein metric on probability-valued global elements is the leading case -- and replacing aggregation by an operation contractive in that metric yields uniqueness and convergent iteration via Banach's theorem on global elements; the $\ell^{1}$ sufficient condition of Remark~\ref{rem:ell-one-contraction} already specializes cleanly.  Stochastic clearing, risk pooling, and capital buffers fit this template, as do non-financial settings in which contraction reflects damping or attenuation rather than risk.  Higher integrability and rates of convergence, statistical inference on partially observed networks, and the interplay between metric contraction and order-theoretic existence are natural sequels.

\style{Cohomological obstructions.}  We anticipate that relative $H^{0}$ and $H^{1}$ of suitable inclusion or boundary maps will function as obstruction classes for \style{partial clearing} -- detecting and measuring not just whether clearing fails, but how it fails, where the failure localizes, and whether locally cleared subsystems extend coherently to global ones; hence the use of cohomological notation $H^{0}(\hyper;\sheaf{L})$ throughout.  The appropriate coefficient setting -- linearized, abelianized, or genuinely nonabelian -- depends on which failures one wishes to detect, and the development of higher cohomology, together with the long exact sequences such a theory would generate, is left to future work.

\style{Further fixed-point principles and dynamics.}  The Tarski/Banach/Kleene triple captures the principal order- and metric-based fixed-point principles developed here. Topological, convex, and multivalued settings invite the corresponding fixed-point theorems of Schauder, Kakutani, and their descendants, all compatible with the liability-category axioms once appropriate hypotheses on $\datacat$ are imposed.  Dynamic and multi-period clearing, in which the base hypergraph or the coefficient category itself varies over time, sits naturally in this family, as do continuous-time formulations, game-theoretic equilibria of strategically defaulting agents, and the algorithmic and complexity-theoretic questions raised by computing clearing sections in large networks.

\bibliographystyle{amsplain}
\bibliography{CLEARING}

@book{MacLane1998,
  author    = {Mac Lane, Saunders},
  title     = {{Categories for the Working Mathematician}},
  edition   = {2},
  series    = {Graduate Texts in Mathematics},
  volume    = {5},
  publisher = {Springer},
  year      = {1998},
  doi       = {10.1007/978-1-4757-4721-8}
}

@book{Leinster2014,
  author    = {Leinster, Tom},
  title     = {{Basic Category Theory}},
  series    = {Cambridge Studies in Advanced Mathematics},
  volume    = {143},
  publisher = {Cambridge University Press},
  year      = {2014},
  doi       = {10.1017/CBO9781107360068}
}

@article{Tarski1955,
  author  = {Tarski, Alfred},
  title   = {A Lattice-Theoretical Fixpoint Theorem and Its Applications},
  journal = {Pacific Journal of Mathematics},
  volume  = {5},
  number  = {2},
  year    = {1955},
  pages   = {285--309},
  doi     = {10.2140/pjm.1955.5.285}
}

@book{DaveyPriestley2002,
  author    = {Davey, B. A. and Priestley, H. A.},
  title     = {Introduction to Lattices and Order},
  edition   = {2},
  publisher = {Cambridge University Press},
  year      = {2002},
  doi       = {10.1017/CBO9780511809088}
}

@book{Bretto2013,
  author    = {Bretto, Alain},
  title     = {{Hypergraph Theory: an Introduction}},
  publisher = {Springer},
  series    = {Mathematical Engineering},
  year      = {2013},
  doi       = {10.1007/978-3-319-00080-0}
}

@phdthesis{Curry2014,
  author  = {Curry, Justin},
  title   = {{Sheaves, Cosheaves and Applications}},
  school  = {University of Pennsylvania},
  year    = {2014},
  note    = {arXiv:1303.3255}
}

@article{HansenGhrist2019,
  author  = {Hansen, Jakob and Ghrist, Robert},
  title   = {Toward a Spectral Theory of Cellular Sheaves},
  journal = {Journal of Applied and Computational Topology},
  volume  = {3},
  year    = {2019},
  pages   = {315--358},
  doi     = {10.1007/s41468-019-00038-7}
}

@article{HansenGhrist2021,
  author  = {Hansen, Jakob and Ghrist, Robert},
  title   = {Opinion Dynamics on Discourse Sheaves},
  journal = {SIAM Journal on Applied Mathematics},
  year    = {2021},
  volume  = {81},
  number  = {5},
  pages   = {2033--2060},
  doi     = {10.1137/20M1341088}
}

@misc{SumrayHarringtonNanda2024,
  author       = {Sumray, Otto and Harrington, Heather A. and Nanda, Vidit},
  title        = {Quiver {Laplacians} and Feature Selection},
  year         = {2024},
  eprint       = {2404.06993},
  archivePrefix = {arXiv},
  primaryClass = {math.AT},
  note         = {arXiv:2404.06993}
}

@inproceedings{DutaCassaraSilvestriLio2023,
  author    = {Duta, Iulia and Cassar\`{a}, Giulia and Silvestri, Fabrizio and Li\`{o}, Pietro},
  title     = {Sheaf Hypergraph Networks},
  booktitle = {Advances in Neural Information Processing Systems},
  volume    = {36},
  pages     = {12087--12099},
  year      = {2023},
  publisher = {Curran Associates, Inc.},
  note      = {arXiv:2309.17116}
}

@inproceedings{DirectionalSHN2025,
  author       = {Mule, Emanuele and Fiorini, Stefano and Purificato, Antonio and Siciliano, Federico and Coniglio, Stefano and Silvestri, Fabrizio},
  title        = {Directional Sheaf Hypergraph Networks: Unifying Learning on Directed and Undirected Hypergraphs},
  booktitle    = {International Conference on Learning Representations},
  year         = {2026},
  note         = {ICLR 2026; arXiv:2510.04727},
  doi          = {10.48550/arXiv.2510.04727}
}

@article{EisenbergNoe2001,
  author  = {Eisenberg, Larry and Noe, Thomas H.},
  title   = {Systemic Risk in Financial Systems},
  journal = {Management Science},
  year    = {2001},
  volume  = {47},
  number  = {2},
  pages   = {236--249},
  doi     = {10.1287/mnsc.47.2.236.9835}
}

@misc{GhristGouldLopezRiess2025,
  author       = {Ghrist, Robert and Gould, Julian and Lopez, Miguel and Riess, Hans},
  title        = {Clearing Sections of Lattice Liability Networks},
  year         = {2025},
  eprint       = {2503.17836},
  archivePrefix = {arXiv},
  primaryClass = {q-fin.MF},
  note         = {arXiv:2503.17836}
}

@article{RogersVeraart2013,
  author  = {Rogers, L. C. G. and Veraart, Luitgard A. M.},
  title   = {Failure and Rescue in an Interbank Network},
  journal = {Management Science},
  year    = {2013},
  volume  = {59},
  number  = {4},
  pages   = {882--898},
  doi     = {10.1287/mnsc.1120.1569}
}

@techreport{Elsinger2009,
  author      = {Elsinger, Helmut},
  title       = {Financial Networks, Cross Holdings, and Limited Liability},
  institution = {Oesterreichische Nationalbank},
  type        = {Working Paper},
  number      = {156},
  year        = {2009}
}

@article{Feinstein2019,
  author  = {Feinstein, Zachary},
  title   = {Obligations with Physical Delivery in a Multilayered Financial Network},
  journal = {SIAM Journal on Financial Mathematics},
  year    = {2019},
  volume  = {10},
  number  = {4},
  pages   = {877--906},
  doi     = {10.1137/18M1194729}
}

@article{KusnetsovVeraart2019,
  author  = {Kusnetsov, Michael and Veraart, Luitgard A. M.},
  title   = {Interbank Clearing in Financial Networks with Multiple Maturities},
  journal = {SIAM Journal on Financial Mathematics},
  year    = {2019},
  volume  = {10},
  number  = {1},
  pages   = {37--67},
  doi     = {10.1137/18M1180542}
}

@article{BanerjeeBernsteinFeinstein2018,
  author  = {Banerjee, Tathagata and Bernstein, Alex and Feinstein, Zachary},
  title   = {Dynamic Clearing and Contagion in Financial Networks},
  journal = {European Journal of Operational Research},
  year    = {2025},
  volume  = {321},
  number  = {2},
  pages   = {664--675},
  doi     = {10.1016/j.ejor.2024.09.046},
  note    = {arXiv:1801.02091}
}

@article{BanerjeeFeinstein2019,
  author  = {Banerjee, Tathagata and Feinstein, Zachary},
  title   = {Impact of Contingent Payments on Systemic Risk in Financial Networks},
  journal = {Mathematics and Financial Economics},
  year    = {2019},
  volume  = {13},
  number  = {4},
  pages   = {617--636},
  doi     = {10.1007/s11579-019-00239-9}
}

@article{BarrattBoyd2020,
  author  = {Barratt, Shane and Boyd, Stephen},
  title   = {Multi-Period Liability Clearing via Convex Optimal Control},
  journal = {Optimization and Engineering},
  year    = {2022},
  volume  = {24},
  number  = {2},
  pages   = {1387--1409},
  doi     = {10.1007/s11081-022-09737-0},
  note    = {arXiv:2005.09066}
}

@article{CalafioreEtAl2023,
  author  = {Calafiore, Giuseppe C. and Fracastoro, Giulia and Proskurnikov, Anton V.},
  title   = {Clearing Payments in Dynamic Financial Networks},
  journal = {Automatica},
  year    = {2023},
  volume  = {158},
  pages   = {111299},
  doi     = {10.1016/j.automatica.2023.111299}
}

@article{ElliottGolubJackson2014,
  author  = {Elliott, Matthew and Golub, Benjamin and Jackson, Matthew O.},
  title   = {Financial Networks and Contagion},
  journal = {American Economic Review},
  year    = {2014},
  volume  = {104},
  number  = {10},
  pages   = {3115--3153},
  doi     = {10.1257/aer.104.10.3115}
}

@article{CapponiChen2015,
  author  = {Capponi, Agostino and Chen, Peng-Chu},
  title   = {Systemic Risk Mitigation in Financial Networks},
  journal = {Journal of Economic Dynamics and Control},
  year    = {2015},
  volume  = {58},
  pages   = {152--166},
  doi     = {10.1016/j.jedc.2015.06.008}
}

@article{AcemogluOzdaglarTahbazSalehi2015,
  author  = {Acemoglu, Daron and Ozdaglar, Asuman and Tahbaz-Salehi, Alireza},
  title   = {Systemic Risk and Stability in Financial Networks},
  journal = {American Economic Review},
  year    = {2015},
  volume  = {105},
  number  = {2},
  pages   = {564--608},
  doi     = {10.1257/aer.20130456}
}

@article{GaiHaldaneKapadia2011,
  author  = {Gai, Prasanna and Haldane, Andrew and Kapadia, Sujit},
  title   = {Complexity, Concentration and Contagion},
  journal = {Journal of Monetary Economics},
  year    = {2011},
  volume  = {58},
  number  = {5},
  pages   = {453--470},
  doi     = {10.1016/j.jmoneco.2011.05.005}
}

@article{GlassermanYoung2015,
  author  = {Glasserman, Paul and Young, H. Peyton},
  title   = {How Likely Is Contagion in Financial Networks?},
  journal = {Journal of Banking \& Finance},
  year    = {2015},
  volume  = {50},
  pages   = {383--399},
  doi     = {10.1016/j.jbankfin.2014.02.006}
}

@article{AnandCraigVonPeter2015,
  author  = {Anand, Kartik and Craig, Ben and von Peter, Goetz},
  title   = {Filling in the Blanks: Network Structure and Interbank Contagion},
  journal = {Quantitative Finance},
  year    = {2015},
  volume  = {15},
  number  = {4},
  pages   = {625--636},
  doi     = {10.1080/14697688.2014.968195}
}

@article{GandyVeraart2017,
  author  = {Gandy, Axel and Veraart, Luitgard A. M.},
  title   = {A {Bayesian} Methodology for Systemic Risk Assessment in Financial Networks},
  journal = {Management Science},
  year    = {2017},
  volume  = {63},
  number  = {12},
  pages   = {4428--4446},
  doi     = {10.1287/mnsc.2016.2546}
}

@inproceedings{SchuldenzuckerSeukenBattiston2017,
  author    = {Schuldenzucker, Steffen and Seuken, Sven and Battiston, Stefano},
  title     = {Finding Clearing Payments in Financial Networks with Credit Default Swaps is {PPAD}-complete},
  booktitle = {8th Innovations in Theoretical Computer Science Conference (ITCS 2017)},
  series    = {LIPIcs},
  volume    = {67},
  pages     = {32:1--32:20},
  year      = {2017},
  publisher = {Schloss Dagstuhl},
  doi       = {10.4230/LIPIcs.ITCS.2017.32}
}

@inproceedings{BestingHoeferHuth2026,
  author    = {Besting, Leander and Hoefer, Martin and Huth, Lars},
  title     = {Computing {Tarski} Fixed Points in Financial Networks},
  booktitle = {43rd International Symposium on Theoretical Aspects of Computer Science (STACS 2026)},
  series    = {LIPIcs},
  volume    = {364},
  pages     = {14:1--14:18},
  year      = {2026},
  publisher = {Schloss Dagstuhl},
  doi       = {10.4230/LIPIcs.STACS.2026.14}
}

@article{CsokaHerings2024,
  author  = {Cs{\'o}ka, P{\'e}ter and Herings, P. Jean-Jacques},
  title   = {Uniqueness of Clearing Payment Matrices in Financial Networks},
  journal = {Mathematics of Operations Research},
  year    = {2024},
  volume  = {49},
  number  = {1},
  pages   = {232--250},
  doi     = {10.1287/moor.2023.1354}
}

@article{CalafioreFracastoroProskurnikov2024,
  author  = {Calafiore, Giuseppe C. and Fracastoro, Giulia and Proskurnikov, Anton V.},
  title   = {Optimal Clearing Payments in a Financial Contagion Model},
  journal = {SIAM Journal on Financial Mathematics},
  year    = {2024},
  volume  = {15},
  number  = {2},
  pages   = {473--502},
  doi     = {10.1137/22M150294X},
  note    = {arXiv:2103.10872}
}

@misc{AyzenbergGebhartMagaiSolomadin2025,
  author        = {Ayzenberg, Anton and Gebhart, Thomas and Magai, German and Solomadin, Grigory},
  title         = {Sheaf Theory: From Deep Geometry to Deep Learning},
  year          = {2025},
  eprint        = {2502.15476},
  archivePrefix = {arXiv},
  primaryClass  = {math.AT},
  doi           = {10.48550/arXiv.2502.15476},
  note          = {arXiv:2502.15476}
}

\end{document}